\documentclass[twoside,10pt]{article}

\usepackage[T1]{fontenc}
\usepackage{times}
\usepackage{amsmath,amsfonts,amssymb,amsthm,euscript,pifont}
\usepackage{mathrsfs}

\usepackage[colorlinks=true]{hyperref}
\usepackage{fullpage}
\usepackage{xspace}
\usepackage{graphicx}
\usepackage{ upgreek }

\usepackage{graphics}
\usepackage{subfigure}
\usepackage{tikz}
\usepackage{url}
\usepackage{authblk}

\newtheorem{theorem}{Theorem}[section]
\newtheorem{corollary}[theorem]{Corollary}
\newtheorem{definition}[theorem]{Definition}

\newtheorem{lemma}[theorem]{Lemma}
\newtheorem{proposition}[theorem]{Proposition}

\newtheorem{assumption}[theorem]{Assumption}
\newtheorem{myremark}[theorem]{Remark}

\newcommand{\er}{\mathbb{R}}
\newcommand{\WW}{\Omega}
\newcommand{\cW}{{\mathcal{W} }}
\newcommand{\ocW}{{\overline{\mathcal{W} }}}

\newcommand{\Ca}{\mathbbm{A}}
\newcommand{\SCM}{\mathscr{C}}
\newcommand{\cS}{\mathbb S}

\newcommand{\pguess}{{\tau^{\text{\rm lower}}}}

\newcommand{\opguess}{{\tau^{\text{\rm higher}}}}

\newcommand{\dom}{\text{Dom}}
\newcommand{\Of}{{\mathcal O}} 
\newcommand{\Off}{{\mathcal O}\!\!\!\,\!\!{\mathcal O}}
\newcommand{\bfs}{{\bf s}}
\newcommand{\bfw}{{\bf w}}

\newcommand{\bfA}{{\bf A}}
\newcommand{\bfSA}{{\mathcal A}\!\!\!\,\!\!{\mathcal A}}
\newcommand{\bfAunder}{{\bf A}^{\tiny \cW}}
\newcommand{\Aunder}{{A}^{\tiny \cW}}
\newcommand{\bfAupper}{{\bf A}^{\tiny \ocW}}
\newcommand{\Aupper}{A^{\tiny \ocW}}
\newcommand{\bfB}{{\bf B}}
\newcommand{\bfBE}{{\bf {B}}}

\newcommand{\bfEp}{{\bf {E'}}}
\newcommand{\bfE}{{\bf {E}}}

\newcommand{\bfC}{{\bf C}}

\newcommand{\op}{\bar{p}}
\newcommand{\uup}{\underbar{$p$}}
\newcommand{\plol}{p_{\text{lolc}}}

\def\co2{$\text{CO}_{2}$}
\newcommand{\pcar}{{p^{\tiny{\text{\rm\tiny\co2}}}}}

\newcommand{\Pnlt}{{\mathfrak{p}}}
\newcommand{\SSigma}{{\Upsigma}\!\!\!\!\!{\Upsigma}}

\newcommand{\sa}{s^{\ast}}
\newcommand{\bfsa}{{\bf s}^{\ast}}
\usepackage{bbm}
\def\ind{{\color{black}{\mathbbm{1}}}}

\newcommand{\pelec}{p^{\rm{elec}}}
\def\etc{\emph{etc}}

\numberwithin{equation}{section}

\title{Game theory analysis for carbon auction market\\through electricity market coupling}

\author[1]{Mireille Bossy\thanks{mireille.bossy@inria.fr}}
\author[2]{Nadia Ma\"{i}zi\thanks{nadia.maizi@mines-paristech.fr}}
\author[3]{Odile Pourtallier\thanks{odile.pourtallier@inria.fr}}

\affil[1]{TOSCA Laboratory, INRIA Sophia Antipolis -- M\'editerran\'ee, France}
\affil[2]{MINES ParisTech, Centre for Applied Mathematics, France}
\affil[3]{HEPHAISTOS Laboratory, INRIA Sophia Antipolis -- M\'editerran\'ee, France}

\date{\today}

\begin{document}
\maketitle
\begin{abstract}
In this paper, we analyze  Nash equilibria between electricity producers selling their production on an electricity market and buying \co2 emission allowances on an auction carbon market. The producers' strategies integrate the coupling of the two markets via the cost functions of the electricity production. We set out a  clear Nash equilibrium on the power market that can be used to compute equilibrium prices on both markets as well as the related electricity produced and \co2 emissions released.

\medskip
{AMS 2010 Subject Classifications 91A80, 91B26.}
\end{abstract}

\section{Introduction}  

The aim of this paper is to develop analytic tools in order to design a relevant mechanism for carbon markets, where relevant refers to emissions reduction. For this purpose, we focus on electricity producers in a power market linked to a carbon market. The link between markets is established through a market microstructure approach. In this context, where the number of agents is limited, standard game theory applies.
The producers are considered as players behaving on the two financial markets represented here by carbon and electricity.  We establish a Nash equilibrium for this non-cooperative $J$-player game through a coupling mechanism between the two markets.
  
The original idea comes from the French electricity sector, where the spot electricity market is often used to satisfy peak demand. Producers' behavior is demand driven and linked to the maximum level of electricity production. Each producer strives to maximize its market share.  In the meantime, it has to manage the environmental burden associated with its electricity production through a mechanism inspired by the  EU ETS (European Emission Trading System)  framework: each  producer unit of emissions must be counterbalanced by a permit or through the payment of a penalty. Emission permit allocations are simulated through a carbon market that allows the producers to buy allowances at an auction. Our focus on the electricity sector is motivated by its prevalence in the emission share (45\% of the whole emission level worldwide), and the introduction in phase III of the EU ETS of an auction-based allowance allocation mechanism. In the present paper, the design assumptions made on the carbon market aim to foster emissions reduction in the entire electricity sector.

Our approach proposes an original framework for the coupling of bidding strategies
on two markets.

Given a  static elastic demand curve on the electricity market (referring to the time stages in an organized electricity market, mainly day-ahead and intra-day), we solve the local problem (just a single time period of the same length for both markets)  
of establishing a non-cooperative Nash equilibrium for the two coupled markets.
This simplification is justified here, as we aim to raise  the condition under which a carbon market would be a real efficient instrument for carbon mitigation policies. 

This analysis is conducted for
non-continuous and non-strictly monotone supply functions and bidding strategies
on both markets in the complete information framework. 

While literature on applied game theory  to strategic  bidding on power markets   mainly addresses profit maximization (see eg \cite{ChiesaDenicolo2009} with complete information, \cite{Hortacsu2011} with private information, \cite{HortacsuPuller2008} with incomplete information),  our objective function is share maximization.  

The equilibria of the coupled markets are based on the  full characterization of the equilibrium electricity price (on the electricity market alone).  
 We prove the uniqueness of the  price and shares,  for share maximization  whereas, to our knowledge  this property is not established (under our hypotheses) for profit maximization   (see eg  \cite{bossy-maizi-al-05}).

Moreover, share maximization approach deals with profit by making specific assumptions, i.e. no-loss sales, and a tradeoff  between the purchase of allowances and the carbon footprint of the electricity generated. Hence,  this work is the first attempt on  power and carbon markets coupling through game theory approach.  Other coupling  approaches use,  for instance, models that produce dynamics for both electricity and carbon prices jointly, as in   \cite{carmona-coulon-schwarz-12}, \cite{carmona-delarue-etal-13}. 

In Section \ref{sec:market-rules}, we formalize the market (carbon and electricity) rules and the associated admissible set of players' coupled strategies.

We start by studying (in section \ref{sec:power-market}) the set of  Nash equilibria  on the electricity market alone (see Proposition \ref{propo-Nash}). This set  constitutes an equivalence class (same prices and market shares) from which we exhibit a dominant  strategy. 

Section \ref{sec:design} is devoted to the analysis of coupled markets equilibria: given a specific carbon market design (in terms of penalty level and allowances), we compute the bounds of the interval where carbon prices  (derived from the previous dominant strategy) evolve. 
We specify the properties of the associated equilibria.

\section{Coupling markets mechanism}\label{sec:market-rules}
\subsection{Electricity market}
In the electricity market, demand is  aggregated and summarized by a function  $p\mapsto D(p)$, where $D(p)$ is the quantity of electricity that buyers are ready to obtain at maximal unit price  $p$.  We assume the following:
\begin{assumption}\label{ass:demande} 
The demand function $D(\cdot):\er^{+} \rightarrow \er^{+}$  is non-increasing, left continuous, and such that  $D(0) >0$.
\end{assumption}

\begin{figure}[ht]
\centering
\subfigure[{\small{delivery 9-10 am}}\label{Clearing1-Epex}]{
\includegraphics[width=.48\textwidth]{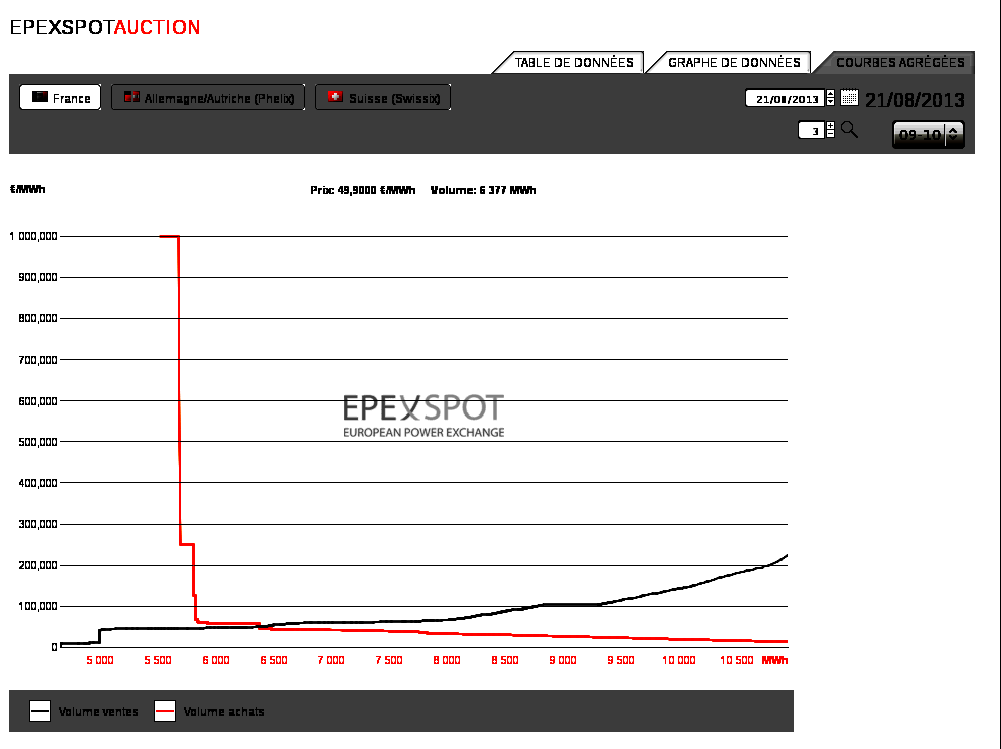}}
\subfigure[{\small{delivery 3-4 pm}}\label{Clearing2-Epex}]{
\includegraphics[width=.48\textwidth]{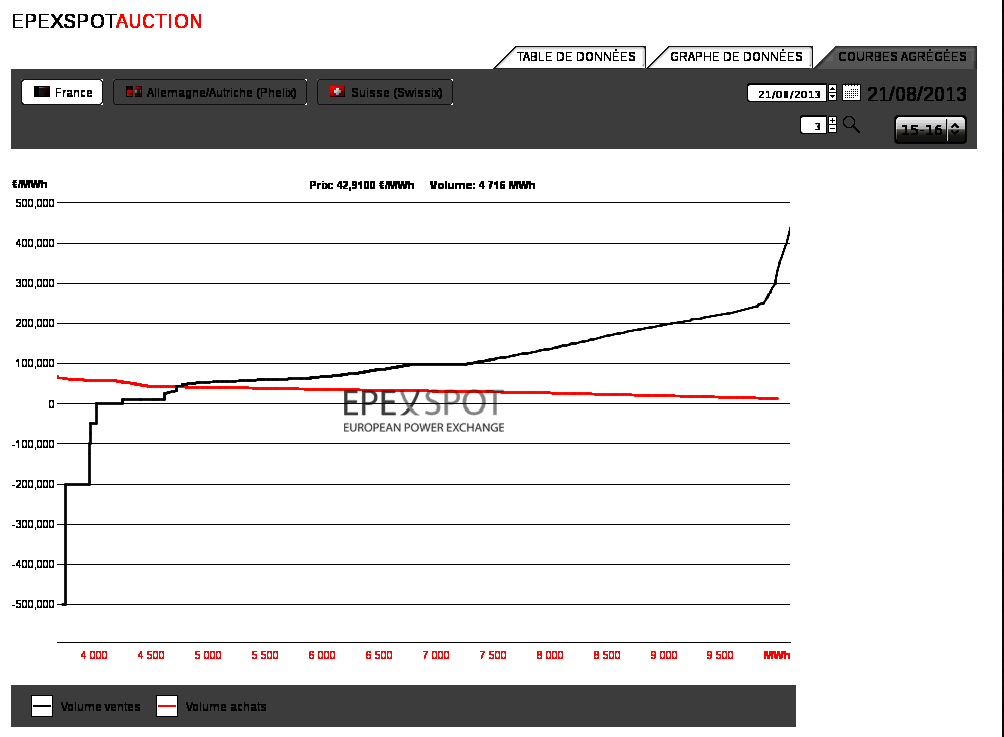}}
\caption{The orange curve is the function $q \mapsto D^{-1}(q)$ on the EPEX market. The evolution of the spot market confirms the relevance of Assumption \ref{ass:demande} on the Demand function $p \mapsto D(p)$.}
\end{figure}

Each producer $j \in \{1, \ldots, J \}$ is characterized by a finite production capacity $\kappa_j$ and a 
bounded and non-decreasing function 
$ c_{j}: [0,\kappa_{j}] \longrightarrow \er^{+}$ that associates a  marginal production cost to any quantity $q$ of 
electricity. These marginal production costs depend on several exogenous parameters 
reflecting the technical costs associated with electricity production e.g. energy prices, O\&M costs, taxes, carbon penalties, \etc. This parameter dependency makes it possible to build different market coupling mechanisms. In the following we use it to link the carbon and electricity markets. 

The merit order ranking features marginal cost functions sorted according to their production costs. These are therefore non-decreasing step functions whereby each step refers to the marginal production cost of a specific unit owned by the producer.

The producers trade their electricity on a dedicated market. For a given producer $j$, the strategy  consists of a function that makes it possible to establish an asking price on the electricity market, defined as
\begin{align*}
s_{j} : & \SCM_j \times \er^{+} \longrightarrow \er^{+} \\
& (c_{j}(\cdot), q) \longrightarrow s_{j}(c_{j}(\cdot), q), 
\end{align*}
where $\SCM_j$ the set of  marginal production cost functions  are explicitly given in the following (see \eqref{def:set-C_j}); 
$s_{j}(c_{j}(\cdot), q)$ is the unit price at which the producer is ready to sell quantity  $q$ of electricity. An admissible strategy  carries out the following sell at no loss constraint
\begin{equation}\label{contrainteStrategie}
s_{j}(c_{j}(\cdot), q) \geq c_{j}(q), \quad \forall q \in  \dom(c_{j}).
\end{equation}
A possible example of such strategy  is $s_{j}(c_{j}(\cdot), q) = c_{j}(q)$ or $s_{j}(c_{j}(\cdot), q) = c_{j}(q)+ \lambda(q)$, where $\lambda(q)$ stands for any additional  profit. 

The constraint \eqref{contrainteStrategie} guarantees profitable trade and  incorporates  an aspect of profit maximization (ie, loss avoidance) into the market share approach. In what follows, we include this profit constraint in the considered class of admissible strategies. 
 
We define the class of admissible strategy profiles on electricity market $\cS$ as: 
\begin{equation}\label{eq:classeStratAdmiss}
\begin{aligned}
\cS = \left \{ 
\begin{array}{l}
\begin{array}{rcl}
\bfs = (s_1,\ldots,s_J); \;  s_{j}: \SCM_j \times \er^{+} & 
\longrightarrow & 
\er^{+} \\ 
(c_{j}(\cdot), q) &
\longrightarrow &
s_{j}(c_{j}(\cdot), q)
\end{array} \\ 
\begin{array}{l}
\mbox{ such that }s_{j}(c_{j}(\cdot), q) \geq c_{j}(q), \quad \forall q \in \dom(c_{j})
\end{array}
\end{array}
\right\}.
\end{aligned}
\end{equation}
As a function of $q$, $s_{j}(c_{j}(\cdot),q)$ is bounded on $\dom(c_{j})$. 
For the sake of clarity, we define for each $q \not \in \dom(c_{j}) $, $s_{j}(c_{j}(\cdot),q) = \plol$, where $\plol$ is the loss of load cost, chosen as any overestimation of the  maximal production costs.

For producer $j$'s  strategy  $s_{j}$,  we define the associated asking size at price   $p$ as 
\begin{equation}\label{defOffrej}
\Of(c_{j}(\cdot),s_{j};p) := \sup\{q, \; s_{j} (c_{j}(\cdot), q) < p \}
\end{equation}
with $\sup \emptyset = 0$. 
Hence $\Of(c_{j}(\cdot),s_{j};p)$ is the maximum quantity of electricity at unit price $p$  supplied by producer $j$ on the market. We also call $p \mapsto \Of(c_{j}(\cdot),s_{j};p)$ the offer function of producer $j$. 

\begin{myremark}\label{property:offre croissante}
\item[(i)] The asking size function  $p\mapsto \Of(c_{j}(\cdot),s_{j};p)$  is, with respect to $p$,  an non-decreasing surjection from   $[0,+\infty)$ to $[0,\kappa_j]$, right
continuous and such that $\Of(c_{j}(\cdot),s_{j};0)=0$.  
For a non-decreasing  strategy  $s_{j}$, $\Of(c_{j}(\cdot),s_{j};.)$ is its generalized inverse function with respect to $q$.
\item[(ii)] Given two strategies  $q\mapsto s_{j}(c_{j}(\cdot), q)$ and  $q\mapsto s_j'(c_{j}(\cdot), q)$  such that $s_{j}(c_{j}(\cdot), q) \leq 
s_j'(c_{j}(\cdot), q)$,  for all $q\in \dom(c_{j})$ 
we have for any positive $p$
\begin{equation*}
\Of(c_{j}(\cdot),s_{j};p) \geq \Of(c_{j}(\cdot),s_j';p).
\end{equation*}
Indeed, if $p_{1} \geq p_{2}$ then $\{ q, \; s_{j}(c_{j}(\cdot), q) \leq p_{2} \} \subset \{ q, \; s_{j}(c_{j}(\cdot), q) \leq p_{1} \}$  from which we deduce that $\Of(c_{j}(\cdot),s_j;\cdot)$ is
non-decreasing. 
Next, if $s_j(c_{j}(\cdot),\cdot) \leq s_j'(c_{j}(\cdot),\cdot)$, for any fixed $p$, we have  
$\{ q, \; s_{j}'(c_{j}(\cdot),q) \leq p \} \subset \{ q, \; s_{j}(c_{j}(\cdot), q) \leq p \}$ from which 
the reverse order follows for the requests. 
\end{myremark}

We shall now describe the electricity market clearing. Note that from a market view point, the dependency of the supply with respect to the marginal cost does not need to be explicit. For the sake of clarity, 
we  write $s_{j}(q)$ and $\Of(s_j;p)$ instead of $s_{j}(c_{j}(\cdot),q)$ and  $\Of(c_{j}(\cdot),s_j;p)$. The dependency will be expressed explicitly whenever needed.

By aggregating the  $J$  asking size functions, we can define the overall asking function $p\mapsto \Off(\bfs;p)$ a producer strategy 
profile $\bfs= (s_{1}, \ldots, s_{J})$ as:
\begin{equation}
\Off(\bfs; p) = \sum_{j = 1}^{J} \Of(s_{j};p).
\end{equation}
Hence, for any producer strategy profile  $\bfs$, $\Off(\bfs ; p)$ is the quantity of electricity that can be sold on 
the market at unit price  $p$.  

The overall supply function $p\mapsto \Off(\bfs; p)$ is a non-decreasing surjection defined from  $[0,+\infty)$ to $[0,\sum_{j=1}^J\kappa_j]$,
such that $\Off(\bfs;0)=0$. 

\subsubsection{Electricity market clearing}
Taking producer strategy profile $\bfs= (s_{1}(\cdot), \ldots, s_{J}(\cdot))$ the market sets the electricity market
price $\pelec(\bfs)$ together with the quantities $(\varphi_{1}(\bfs), \ldots, \varphi_J(\bfs))$ of electricity sold by each producer. 

The market clearing price $\pelec(\bfs)$ is the unit price paid to each producer for the quantities $\varphi_{j}(\bfs)$ of  electricity. The price  $p(\bfs)$ may be defined as a price whereby supply satisfies demand. As  we are working with a general non-increasing demand curve (possibly locally inelastic), the price that satisfies the demand is not necessarily unique.  We thus define the clearing price generically with the following definition.  
\begin{definition}[The clearing electricity price]\label{def:clearingElec}
Let us define
\begin{equation}\label{reacmarche-prix}
\begin{aligned}
&\uup(\bfs) = \inf \left\{ p > 0 ; \;  \Off(\bfs ; p) >  D(p) \right\} \\
\mbox{ and } \quad &
\op(\bfs) = \sup \left\{ p\in [\uup(\bfs),\plol];  D(p) = D(\uup(\bfs))\right\}
\end{aligned}
\end{equation}
with the convention that $\inf\emptyset = \plol$. The clearing price may then be established as any  $\pelec(\bfs) \in [\uup(\bfs), \op(\bfs)]$ as an output of a specific market clearing rule. To keep the price consistency, the market rule  must be such that for any two strategy profiles $\bfs$ and $\bfs '$, 
\begin{equation}\label{regleChoixPrix}
\begin{aligned}
\mbox{if } \uup({\bfs}) < \uup({\bfs '}) \mbox{ then } \pelec({\bfs}) <  \pelec({\bfs '}), \\
\mbox{if } \uup({\bfs}) = \uup({\bfs '}) \mbox{ then } \pelec({\bfs}) =  \pelec({\bfs '}). 
\end{aligned}   
\end{equation}
\end{definition}

Note that $\uup(\bfs)\neq \op(\bfs)$ only if the demand curve $p\mapsto D(p)$ is constant on some intervals $[\uup(\bfs),\uup(\bfs)+ \epsilon]$. In that case, 
$\uup(\bfs)$ corresponds to the best ask price, while $\op(\bfs)$ is the best bid price. 
The demand/offer curves that result from the buyer/seller aggregation in a given  time period implies  some market fixing rules that allocate  buyer surplus and seller surplus. In that sense $\pelec({\bfs})$ is a fixing price\footnote{One can imagine that  Power market  participants have access to the detailed fixing rules, but information proves hard to be found.}.  
Note that $\pelec({\bfs})=\uup(\bfs)$ maximizes buyer surplus while $\pelec({\bfs})=\op(\bfs)$  maximizes seller surplus. 

\begin{figure}[ht]
\begin{center}
\begin{tikzpicture}[xscale=10,yscale=0.03]\footnotesize
 \newcommand{\xone}{-.02}
 \newcommand{\xtwo}{ 1.04}
 \newcommand{\yone}{-.4}
 \newcommand{\ytwo}{185}

\begin{scope}<+->;
\draw[black] (0,0) node[anchor=north east] {$0$};
\draw[black,thick,->] (\xone, 0) -- (\xtwo, 0);
\draw[black,thick,] (0.95, -15)  node[right] {price};
\draw[black,thick,->] (0, \yone) -- (0, \ytwo)node[left] {quantity};
\end{scope}
\begin{scope}[thick,blue]
\draw (0.0,0) node {$\bullet$} ;
\draw (0.0,0) -- (0.1,0);
\draw (0.1,25) node {$\bullet$} ;
\draw (0.1,25) -- (0.2,25);
\filldraw[very thin,opacity=.2] (0.1,0.0) rectangle (0.2,25);

\draw (0.2,40) node {$\bullet$} ;
\draw (0.2,40) -- (0.4,40);
\filldraw[very thin,opacity=.2] (0.2,0) rectangle (0.4,40);

\draw (0.4,50) node {$\bullet$} ;
\draw (0.4,50) -- (0.55,50);
\filldraw[very thin,opacity=.2] (0.4,0) rectangle (0.55,50);

\draw (0.55,100) node {$\bullet$} ;
\draw (0.55,100) -- (0.65,100);
\filldraw[very thin,opacity=.2] (0.55,0) rectangle (0.65,100);
\draw (0.59,80) node[right] {Total offer $p\mapsto \Off(p)$};

\draw (0.65,120) node {$\bullet$} ;
\draw (0.65,120) -- (0.9,120);
\filldraw[very thin,opacity=.2] (0.65,0) rectangle (0.9,120);

\draw (0.9,150) node {$\bullet$} ;
\draw (0.9,150) -- (1.0,150);
\filldraw[very thin,opacity=.2] (0.9,0) rectangle (1.0,150);
\end{scope}

\begin{scope}[thick,red]
\draw (0.17,100) node[right,above] {Demand $p \mapsto D(p)$} ;
\draw (0.15,175) node {$\bullet$} ;
\draw (0,175) -- (0.15,175);
\filldraw[very thin,opacity=.2] (0.0,0.0) rectangle (0.15,175);

\draw (0.27,143) node {$\bullet$} ;
\draw (0.15,143) -- (0.27,143);
\filldraw[very thin,opacity=.2] (0.15,0) rectangle (0.27,143);

\draw (0.32,130) node {$\bullet$} ;
\draw (0.27,130) -- (0.32,130);
\filldraw[very thin,opacity=.2] (0.27,0) rectangle (0.32,130);

\draw (0.7,65) node {$\bullet$} ;
\draw (0.32,65) -- (0.7,65);
\filldraw[very thin,opacity=.2] (0.32,0) rectangle (0.7,65);

\draw (0.9,32) node {$\bullet$} ;
\draw (0.7,32) -- (0.9,32);
\filldraw[very thin,opacity=.2] (0.7,0) rectangle (0.9,32);

\draw (0.9,19) -- (1,19);
\filldraw[very thin,opacity=.2] (0.9,0) rectangle (1,19);
\end{scope}

\begin{scope}[black]
\draw[dashed] (0.55,0.0) -- (0.55,150);
\draw (0.55,0.0) node {$\bullet$} ; 
\draw[thin,<-]  (0.55,-4) -- (0.45,-14) node[below] {$\uup(\bfs)$};

\draw[dashed] (0.7,0.0) -- (0.7,150);
\draw (0.7,0.0) node {$\bullet$} ; 
\draw[thin,<-]  (0.7,-4) -- (0.80,-14) node[below] {$\op(\bfs)$};

\draw[dashed]  (0.32,65)  --  (-0.0,65)  node {$\bullet$} node[left]{quantity sold};
\end{scope}
\end{tikzpicture}
\end{center}
\caption{Electricity  price $\uup(\bfs)$ and $\op(\bfs)$.\label{clearing} }
\end{figure}

Note also that price $\uup(\bfs)$ is well defined in the case 
where demand does not strictly decrease. This includes the case where demand is constant. 
In such case, $\uup(\bfs)=\plol$ only if the demand curve never crosses the supply. 

Next, we define the quantity of electricity sold at price $\pelec(\bfs)$. 
When $\pelec(\bfs)$ is such that $\Off(\bfs;\pelec(\bfs)) \leq D(\pelec(\bfs))$, each producer sells $\Of(\bfs_{j};\pelec(\bfs))$, but cases where  $\Off(\bfs;\pelec(\bfs)) >  D(\pelec(\bfs))$ may occur, requiring the introduction of an auxiliary rule to share $D(\pelec(\bfs))$ among the producers that propose $\Off(\bfs;\pelec(\bfs))$. 
Note that in this last case, due to the clearing property \eqref{regleChoixPrix} on $\pelec(\cdot)$, we have 
$
\Off(\bfs;\uup(\bfs)) > D(\pelec(\bfs))= D(\uup(\bfs)).
$
Hence the $D(\pelec(\bfs))$ is totally provided by producers with non null offer at price $\uup(\bfs)$. The rule of the market is to share $D(\pelec(\bfs))$  among these producers only. This gives an explicit  priority to the best offer prices $\uup(\bfs)$. 

Let us break down supply as follows: 
\begin{align*}
\Off(\bfs ; \uup(\bfs)) = \sum_{j=1}^J \Of(s_j;\uup(\bfs)^-) + \sum_{j=1}^J \Delta^- \Of(s_j;\uup(\bfs)), 
\end{align*}
where
$ \Delta^- \Of(s_j;\uup(\bfs)) := \Of(s_{j};\uup(\bfs)) - \Of(s_{j}; \uup(\bfs)^{-})$ and $f(x^{-})$ denotes the left value at $x$ of a function $f$.

The market's choice is to fully accept the asking size of producers with continuous asking size curve at point $\uup(\bfs)$. For producers with discontinuous asking size curve at $\uup(\bfs)$, a market rule based on proportionality that favors abundance
is used to share the remaining part of the supply: any extra supply available at the clearing price $\Off(\bfs ; \uup(\bfs)) - D(\uup(\bfs))$
is split among all generators offering at that price such that they each get the same percentage
of their offered quantity allocated to production. 

We summarize the market rule on quantities as follows. 
\begin{definition}[Clearing electricity quantities]\label{def:quantities}
The  quantity $\varphi_{j}(\bfs)$ of electricity sold by Producer $j$ on the electricity market  is 
\begin{equation}\label{reacmarche-quantite-elec}
\begin{aligned}
\quad \varphi_{j}(\bfs) = \left\{
\begin{array}{l}
\Of(\bfs_{j}; \pelec(\bfs)) { = \Of(\bfs_{j}; \uup(\bfs))},
\quad\mbox{ if }D(\pelec(\bfs))\geq \Off(\bfs ; \pelec(\bfs)),  \\ \\ 
\Of(\bfs_{j}; \uup(\bfs)^-)+\Delta^-\Of(\bfs_{j};\uup(\bfs))\dfrac{D(\uup(\bfs)) - \Off(\bfs ; \uup(\bfs)^{-})}{\displaystyle \Delta^- \Off(\bfs ; \uup(\bfs))},\\
\quad \quad\quad\mbox{ if }D(\uup(\bfs)) < \Off(\bfs ; \uup(\bfs)),
\end{array}\right.
\end{aligned}
\end{equation}
where
$ \Delta^- \Off(\bfs ; \uup(\bfs)) := \displaystyle\sum_{j=1}^{J}  \Delta^- \Of(s_j;\uup(\bfs)) > 0$. 
\end{definition}
Note that, when $D(\uup(\bfs)) < \Off(\bfs ; \uup(\bfs))$, we have $\Delta^- \Off(\bfs ; \uup(\bfs)) > 0$.  
Note also that 
\begin{equation}\label{eq:clearingQuantities}
\begin{aligned}
\sum_{i=1}^J \varphi_{j}(\bfs)  & = D(\pelec(\bfs))\wedge \Off(\bfs ; \pelec(\bfs))
  = D(\uup(\bfs))\wedge \Off(\bfs ; \uup(\bfs)), 
\end{aligned}
\end{equation}
\begin{equation}\label{eq:clearingQuantity}
\mbox{and for all } j, \quad  \Of(s_j;\uup(\bfs)^-) \leq \varphi_j(\bfs) \leq \Of(s_j;\uup(\bfs)).
\end{equation}

\subsection{Carbon market}

Let us recall the \co2 regulation principle on which we base our analysis. 
Producers are penalized according to their emission level if they do not own allowances. Hence, in parallel 
 to their position on the  electricity market,  producers  buy \co2 emission allowances on a separate \co2 auction market. 
In the following, we formalize producer strategy  on the \co2 market only. 

If they are allowed to, producers buying permits on the \co2 market will use them 
either to cap their own power production emissions, either to prevent other players from  buying permits.  The following assumption introduces some market design rules that control  players behavior  on that market.

\begin{assumption}[Capped carbon market]\label{ass:co2Cap}
\item[(i)] The carbon market  is capped and has a finite known quantity $\WW$ of \co2  emission allowances available. 

\item[(ii)] Each producer $j$ can buy a capped number of allowances $\mathcal{E}_{j}$, related to its own \co2  emission capacity. 

\item[(iii)] Emissions that are not covered by allowances are penalized at a unit rate $\Pnlt$. 
\end{assumption}

{Note that if one chose $\mathcal{E}_{j}\geq  \WW$ for all producers then item {\it{(ii)}} is void.  Other choice for the $\mathcal{E}_{j}$ can  be seen as strengthen regulation tool. }

On this market, producers adopt a strategy that consists of an offer function $\tau \mapsto A_j(\tau)$ defined from  $[0,\Pnlt]$ to $[0,\mathcal{E}_j]$. Quantity $A_{j}(\tau)$  is the quantity of allowances that producer $j$ is ready to buy at price $\tau$.  This offer may not be a monotonic function.  We denote $\Ca$ the strategy profile set on the \co2 market,
\[
\Ca := \{ \bfA = (A_{1},\ldots ,A_{J}); \mbox{ s.t.   }A_j:[0,\Pnlt]\rightarrow[0,\mathcal{E}_j]\}.
\]

The \co2 market reacts by aggregating the $J$ offers by
\[\bfSA (\tau) := \sum_{j=1}^J A_j(\tau),\] 
and the clearing market price is established following a {\it second item auction}\footnote{
Also called {\it Dutch auction market} with several units to sell,   in a {\it second item auction} market, the seller begins with a very high price and reduces it. The price is lowered until a bidder accepts the current price.}  as:

\begin{align}\label{reacmarche-prix-quotas}
\pcar(\bfA) := \sup\{\tau; \bfSA(\tau) > \WW\}, \quad \mbox{ with the convention }
\sup \emptyset = 0. 
\end{align}

Note that $\pcar(\bfA) =0$ indicates that there are too many allowances to sell. It is worth a reminder here that the aim of allowances is to decrease emissions. In section \ref{sec:design}, we discuss a design hypothesis (Assumption~\ref{ass:hypoTW}) that guarantees an equilibrium price $\pcar(\bfA) >0$.
Therefore, in the following, we assume that the overall quantity $\WW$ of allowances, is such that $\pcar(\bfA) >0$.

\medskip
Next, we define the amount of allowances bought at price $\pcar(\bfA)$ by the producers. 
By Definition \eqref{reacmarche-prix-quotas}, we have 
$\bfSA( \pcar(\bfA))\geq \WW ~\mbox{ and }~ \bfSA({\pcar}(\bfA)^{+})\leq \WW$. When $\bfSA( \pcar(\bfA))> \WW$, the \co2 market must decide between the producers with an additional rule. We define  
\[\Delta(A_{i}) := A_i(\pcar(\bfA) ^+) - A_i(\pcar(\bfA)).\]

For a producer $i$, $\Delta(A_{i}) \geq 0$ means that its \co2 demand does not decrease if the price increases. It is therefore ready to pay more to obtain the quantity of allowances it is asking for at price $\pcar(\bfA)$.  The \co2 market gives priority to this kind of producer, which will be fully served.  
The producers such that $\Delta(A_{j}) < 0$ share the remaining allowances. This can be written as follows. 

Each producers with  $A_j(\pcar(\bfA))>0$ obtains the following quantity $\delta_j(\bfA)$ of  allowances 
\begin{equation}\label{reacmarche-quantite-quotas}
\delta_{j}(\bfA) := 
\left\{ 
\begin{array}{l}
A_{j}(\pcar(\bfA)), \quad  \mbox{ if } \Delta(A_{j})  \geq 0 ,\\ \\
A_j(\pcar(\bfA) ^+) 
 +   \frac{\displaystyle(-\Delta(A_{j}))^+}{\displaystyle\sum_{i=1}^{J} (-\Delta(A_{i}))^{+} }\left(\WW -  \displaystyle \sum_{i=1}^J  A_{i}(\pcar(\bfA)) \ind_{\{\Delta(A_{i})  \geq 0 \}}  \right) 
,\\
\quad \mbox{ otherwise.}
\end{array}\right.
\end{equation}

\subsection{Carbon and electricity market coupling}

In the following, we formalize the coordination of a producer's strategy  on the \co2 and electricity markets.  This could be seen as if both markets were synchronized during a single time period with the same length (eg, one hour).

As mentioned earlier, for each producer, the marginal cost function is parametrized by the positions 
$\bfA$ of the producers  on the carbon market. Indeed, producer $j$ can obtain  \co2 emission 
allowances on the market to avoid penalization for  (some of) its emissions. Those emissions that are not covered by allowances are penalized at a unit rate $\Pnlt$. 

A profile of an offer to buy from the producers  $\bfA = (A_1, \ldots, A_J)$, through the \co2 market clearing, corresponds to a 
unit price of $\pcar(\bfA)$ of the allowance and quantities $\delta_{j}(\bfA)$ of allowances  bought by each producer (defined by the market rules  ~\eqref{reacmarche-prix-quotas},\eqref{reacmarche-quantite-quotas}). 

This yields to the following modified marginal production cost function\footnote{Note that this representation might also include the allowances possibly stored  by the producers in the previous periods.}   
 $c_{j}^\bfA(\cdot)$, parametrized by the emission regulations: 
\begin{equation}\label{coutsRegulation}
q\mapsto c_{j}^\bfA(q) = \left \{ 
\begin{array}{ll}
c_j(q) + {e_j}(q) \pcar(\bfA),&  \mbox{ for  }  q\in [0,\kappa^{\text{\tiny \co2}}_j (\bfA)\wedge \kappa_j] \\
c_j(q) + {e_j}(q) \Pnlt, & \mbox{ for  }  q\in [\kappa^{\text{\tiny \co2}}_j (\bfA)\wedge \kappa_j, \kappa_j] 
\end{array}
\right.
\end{equation}
where  for all producers $\{j=1,\ldots, J\}$,  
\begin{itemize}
\item $q\mapsto e_{j}(q)$  is the emission rate (originally in \co2 t/Mwh),   

\item $\kappa^{\text{\tiny \co2}}_j (\bfA)$ is  the electricity capacity covered by the bought allowances $\delta_{j} (\bfA) \leq \mathcal{E}_{j}$:  
\[\kappa^{\text{\tiny \co2}}_j (\bfA) = \text{argmax} \{k;  \int_0^{k} e_j(z) dz \leq  \delta_{j} (\bfA)\}.\]
\end{itemize}

In this coupled market  setting, the strategy of producer  $j$ thus makes a pair $(A_{j}, s_{j})$. The set of admissible strategy profile is defined as 
\begin{align}\label{def:setOfAdmissibleStrategies}
\SSigma = \left \{ (\bfA,\bfs); \;\bfA\in\Ca, \;\bfs\in \cS \right\}, 
\end{align}
where in the definition of $\cS$ in \eqref{eq:classeStratAdmiss}, we use  
\begin{align}\label{def:set-C_j}
\SCM_j = \left\{ c^\bfA_j; \;\bfA \in \Ca \right\}. 
\end{align}
Prices for allowances and electricity, $\pcar({(\bfA,\bfs)})$ and $\pelec({(\bfA,\bfs)})$, quantities of allowances bought by each producer, $\delta_{j}({(\bfA,\bfs)})$ and market shares on electricity market $\varphi_{j}({(\bfA,\bfs)})$ of each producer corresponds to any strategy profile ${(\bfA,\bfs)} \in \SSigma$, through the market mechanisms described.

\section{Nash Equilibrium analysis}\label{sec:nash}

We suppose that the $J$ producers behave non-cooperatively, aiming at maximizing their individual market share on the electricity market. For a strategy profile ${(\bfA,\bfs)} \in \SSigma$, the market share of a producer $j$ depends upon its strategy 
$(A_{j},s_{j}(\cdot))$ but also on the strategies $(\bfA_{-j},\bfs_{-j})$ of the other producers\footnote{Here we use the generic notation ${\boldsymbol{b}}_{-j}$ that  stands for
the profile set $(b_{1},\cdots, b_{j-1},b_{j+1},\cdots, b_{J})$.}. 
In this set-up the natural solution is the Nash equilibrium (see e.g. \cite{basar-olsder-98}).
More precisely we are looking for a strategy profile 
\[
{(\bfA^{\ast},\bfs^{\ast})} = ( (A_{1}^{\ast},s_{1}^{\ast}), \cdots, (A_{J}^{\ast},s_{J}^{\ast}) ) \in \SSigma
\]
that satisfies Nash equilibrium conditions:  none of the 
producers would strictly benefit, that is, would strictly increase its market share from a unilateral deviation. Namely, for any producer $j$ strategy ${(\bfA_{j},\bfs_{j})}$ such that 
$({(\bfA^{\ast}_{-j},\bfs^{\ast}_{-j})}, {(A_j,s_j)}) \in \SSigma $, we have
\begin{align}\label{NashGlob}
\varphi_{j}({{(\bfA^{\ast},\bfs^{\ast})}}) \geq  \varphi_{j}({{(\bfA^{\ast}_{-j},\bfs^{\ast}_{-j})}},{(A_j,s_j)}), 
\end{align}
where $\varphi_j$ is the quantity of electricity sold. Note that the dependency in terms of $\bfA$ through the marginal cost $c_j^\bfA$ is now made explicit in $\varphi_j$. 

Condition \eqref{NashGlob} has to be satisfied for any unilateral deviation of any producer $j$. In particular  \eqref{NashGlob}  has to be satisfied for a producer $j$ 
admissible deviation $(A_{j}^{\ast}, s_{j})$ such that  $({(\bfA^{\ast}_{-j},\bfs^{\ast}_{-j})}, {(A_j^{\ast},s_j)}) \in \SSigma $ where producer $j$ would only change its behavior on the electricity market. 

\begin{myremark}\label{rem:NashPartial}~

The electricity  strategy component  $\bfs ^{\ast}$ of the Nash equilibrium ${(\bfA^{\ast},\bfs^{\ast})} $  is also a Nash equilibrium for the restricted  electricity game, where producers only behave on the  electricity market with  marginal electricity  production costs $c_{j}^{\bfA^{\ast}}(\cdot)$, $j=1, \cdots J$.

\end{myremark}

The next section focuses on determining a Nash equilibrium on the game restricted to the electricity market.

\subsection{Equilibrium on the power market}\label{sec:power-market}

In this restricted set-up, we consider that the marginal costs   $\{c_j,j=1\ldots,J\}$ are known data, possibly fixed through the position $\bfA$ on the \co2 market. In this section, we refer to $\cS$ as the set  of admissible strategy profiles, in the particular case where $\SCM_j=\{c_j\}$ for each $j=1,\ldots,J$. 

The Nash equilibrium problem is as follows: find  $\bfsa = (\sa_{1}, \ldots, \sa_{J}) \in \cS$  such that 
\begin{align}\label{nashQuantiteElec}
\begin{aligned}
\forall j, \forall \;s_{j}\neq \sa_{j}, \quad \varphi_{j}(\bfsa) \geq \varphi_{j}(\bfsa_{-j}, s_{j}).
\end{aligned}
\end{align}

The following proposition exhibits a Nash equilibrium, whereby each producer must choose the strategy denoted by $C_{j}$, and referred to as {\it marginal production cost strategy}. It is defined by 
\begin{align}\label{stratCoutMarg}
C_{j}(q) = \left \{ 
 \begin{array}{l}
   c_{j}(q), \mbox{ for } q \in \dom(c_{j}) \\
   \plol , \mbox{ for } q \not \in \dom(c_{j}) .
 \end{array} \right.
\end{align}
\begin{proposition}\label{propo-Nash}
\item[(i)] 
For any strategy profile ${\bfs} = (s_1,\ldots,s_J)$, no producer  $j\in\{1,\ldots,J\}$ can
be penalized by deviating from strategy  $s_j$ to its marginal production cost strategy $C_j$, namely, 
\begin{equation}\label{propoPti}
\varphi_j(\bfs)\leq \varphi_{j}({\bf{s}}_{-j},C_j).
\end{equation}
In other words, for any producer $j$, $C_{j}$ is a dominant strategy.
\item[(ii)] The strategy profile $\bfC =(C_{1},\dots C_{J})$ is a Nash equilibrium.
\item[(iii)] If the strategy profile $\bfs \in \cS$ is a Nash equilibrium, then we have 
 $\op(\bfs) = \op(\bfC)$,    and for any producer $j$, $\varphi_{j}(\bfs) = \varphi_{j}(\bfC)$.
\end{proposition}

Point \textit{(ii)}
exhibits a Nash equilibrium strategy profile as a direct consequence of the dominance property \textit{(i)}. Clearly the Nash equilibrium is non-unique, since we can easily show that a producer's given supply can follow from countless different strategies. Nevertheless point \textit{(iii)} shows that there is a unique associated  quantities of electricity sold by producers.
The market coupling mechanism that we propose in the following section is based on this uniqueness property which allows the computation of the equilibrium shares on electricity and carbon markets.
Moreover, any Nash equilibrium price evolves in the interval $\pelec(\bfs) \in [\uup(\bfC), \op(\bfC)]$, which reduces to the point $\{\op(\bfC)\}$ in various situations, in particular when $D(\cdot)$ strictly decreases at $\uup(\bfC)$, or when $\pelec$ is chosen equal to $\op$.

Proofs of   \textit{(i)} and \textit{(iii)}, which are rather tedious due to  non-strictly monotony and  possible discontinuity of  supply and offers,    are postponed  to  Appendix \ref{appendix:proofNash}.

\subsection{Coupled market design using the Nash equilibrium}\label{sec:design}

From this point we restrict our attention to a particular market design. In the following, the scope of the analysis applies to a special class of producers, a specific electricity market price clearing (satisfying Definition \ref{def:clearingElec}) and a range of quantities $\WW$ of allowances available on the \co2 market. Although not necessary, the following  restriction simplifies the development. 

\begin{assumption}{\bf On the producers.~}\label{ass:producers}
Each producer $j$ operates a single production unit, for which  
\begin{itemize}
\item[(i)] ~the initial marginal cost contribution (that does not depend on the producer positions $\bfA$ in the \co2 market) is constant,
$q\mapsto c_{j}(q) = c_{j}\index{q \in [0,\kappa_{j}]}$. The related emission rate $q\mapsto e_{j}(q) = e_j$ is also assumed to be a positive constant,

\item[(ii)] ~the producers are different pairwise: $\forall i, j \in \{1, \cdots J \}, (c_{i},e_{i}) \neq (c_{j},e_{j})$.
\end{itemize}
\end{assumption}

In what  follows (according to Assumption \ref{ass:co2Cap}), in order to limit the number of parameters involved in the discussion, the maximal cap of allowances that  each producer $j$ may buy is  set to $\mathcal{E}_{j}=e_j {\kappa_{j}}$. This arbitrary but natural choice does not penalize producers capacity level, and does not bring any  restriction to the following equilibrium analysis. 

As a consequence of Assumption \ref{ass:producers}, the marginal production cost in \eqref{coutsRegulation}  can simply be written as 
\begin{equation}\label{coutsRegulationbis}
q\mapsto c_{j}^\bfA(q) = \left \{ 
\begin{array}{ll}
c_j+ {e_j}\pcar(\bfA),&  \mbox{ for  }  q\in [0, \displaystyle \frac{\delta_{j} (\bfA)}{e_j}  \wedge \kappa_j] \\
c_j+ {e_j} \Pnlt, & \mbox{ for  }  q\in [\displaystyle \frac{\delta_{j} (\bfA)}{e_j} \wedge \kappa_j, \kappa_j].
\end{array}
\right.
\end{equation}

For a given  strategy profile on the electricity market, Definition \ref{def:clearingElec} gives a range of possible determinations for the electricity price. 
Previously, the analysis of the  Nash Equilibrium restricted to the  electricity market did not require a precise clearing price determination. 
Nevertheless to extend our analysis to the coupling we need to make explicit this determination and assume the following:
\begin{assumption}{\bf On the electricity market.~}\label{ass:ElecClearingPrice}
For a given strategy profile $\bfs$ of the producers, the clearing price of electricity is 
$\pelec(\bfs)$. The market rule fixes   $\pelec(\cdot)=\op(\cdot)$ or $\pelec(\cdot)=\uup(\cdot)$  as defined in  \eqref{reacmarche-prix}.
\end{assumption}
We will illustrate below that  this choice of clearing price ensures the increasing behavior of $\pelec(\cdot)$ and right continuity in terms of the carbon price (see Lemma \ref{lem:pelecCroissante}). 

The quantity $\WW$ of \co2 allowances available  plays a crucial role in the market 
design. If this quantity is too high, the allowance's market price will drop to zero, leaving 
the market incapable of fulfilling its role of decreasing \co2 emissions. Therefore we clearly
need to make an assumption that restricts the number of allowances available. 
Appropriately capping the maximum quantity of allowances available requires information on which producers are willing to obtain allowances. This is the objective of the following paragraph where we define a {\it willing to buy} function that plays a central role in the analysis of Nash equilibria.  

\subsubsection{Willing to buy functions}
In this paragraph, we aim at guessing a Nash equilibrium candidate. We base our reasoning on the dominant strategies on the electricity market alone (see Proposition~\ref{propo-Nash}). Remark \ref{rem:NashPartial} allows us to fix the electricity market strategy as a {\it marginal production cost strategy}, given the marginal cost functions $\bfC^\bfA = \{c_{j}^\bfA,j=1,\ldots J\}$ imposed by the output of the \co2 clearing, as in  \eqref{coutsRegulationbis}. 

In particular, when $\bfA\in\Ca$, we observe that the strategies $(\bfA,\{c_{j}^\bfA,j=1,\ldots J\})$ are in the set of admissible strategies defined in \eqref{def:setOfAdmissibleStrategies}. 

\medskip
From now on, all the strategy profiles that we consider on the carbon market are assumed to be admissible. 

In the following, as the discussion will mainly focus on the impact of strategies $\bfA$ through the carbon market, we denote the electricity market output as:
\begin{equation}\label{changeOfNotation}
\begin{aligned}
\pelec(\bfA) \quad & \mbox{ instead of } \quad  \pelec(\bfC^\bfA)\\
(\varphi_1(\bfA),\ldots,\varphi_J(\bfA)) \quad & \mbox{ instead of } \quad  (\varphi_1(\bfC^\bfA),\ldots,\varphi_J(\bfC^\bfA)).
\end{aligned}
\end{equation}

To begin with, we consider an exogenous  \co2 cost $\tau$ similar to a \co2 tax rate:  the producers' marginal cost becomes for any $\tau \in [0,\Pnlt]$, $c_{j}^{\tau}(\cdot)$,
\[c_{j}^{\tau}(q) = c_{j} + \tau e_{j}, \;\mbox{for}\; q \in [0,\kappa_{j}], ~j=1,\ldots,J.\]

In this {\it tax} framework, the dominant strategy on the electricity market is also parametrized by $\tau$ as $\bfC^{\tau}=\{c_{j}^\tau,j=1,\ldots J\}$ defined in \eqref{stratCoutMarg}. The clearing electricity price and quantities follow as
\begin{equation}\label{changeOfNotationTau}
\begin{aligned}
\pelec(\tau) &= \pelec (\bfC^{\tau}),\\
(\varphi_1(\tau),\ldots,\varphi_J(\tau)) &=  (\varphi_1(\bfC^\tau),\ldots,\varphi_J(\bfC^\tau)).
\end{aligned}
\end{equation}
Price $\pelec(\tau)$ will be referred to as the {\it taxed} electricity price, by contrast with price $\pelec({\bfA})$ issued from the {\it marginal production cost strategy}  that results from the position $\bfA$ on the carbon market.
\begin{myremark}\label{rem:diffDePrix}
Considering a carbon tax $\tau$ and a carbon market strategy $\bfA$ such that $\tau = \pcar(\bfA)$, 
we emphasize  the fact that the corresponding electricity prices are not equivalent, but we always have the following inequality 
\[ \pelec(\tau) \leq \pelec(\bfA).\]
This follows from the fact that  for all $i$, $c_i^\tau(\cdot) \leq c_i^A(\cdot)$ and hence $\Of(c_i^\bfA;\cdot) \leq \Of(c_i^\tau;\cdot)$.  The gap between $\bfC^\tau(\cdot)$ and $\bfC^A(\cdot)$ comes  both  from the width ($\Omega$ effect)  and the height (penalty effect) of their steps. 
\end{myremark}

We start with the following:
\begin{lemma}\label{lem:pelecCroissante} 
Under Assumption \ref{ass:ElecClearingPrice}, the map $\tau\mapsto \pelec(\tau)$ is non-decreasing  and right  continuous. 
\end{lemma}
We determine the {\it willing-to-buy-allowances functions}   $\cW_{j}(\cdot)$ and 
$\cW(\cdot)$, as follows: 
\begin{align}\label{WillingQuotaj}
\cW_{j} (\tau) = e_{j} \varphi_j(\tau) 
\quad \mbox{ and }\quad
\cW(\tau) = \sum_{j=1}^J  \cW_{j}(\tau).
\end{align}
For producer $j$, $\cW_{j}$ is the quantity of emissions it would produce under the penalization $\tau$, and consequently the quantity of allowances it would be ready to buy at price $\tau$. Given the \co2 value $\tau$, the total amount $\cW (\tau)$ represents the allowances needed to cover the global emissions generated by  the players who have won electricity market shares.
We also define the functions 
\begin{align}\label{WillingQuotaMaxj}
\ocW_{j}(\tau) = e_{j} \kappa_{j}  \ind_{\{\varphi_j(\tau)>0\}}, 
\quad \mbox{ and }\quad
\ocW(\tau) = \sum_{j=1}^J  \ocW_{j}(\tau).
\end{align}
Given that the \co2 value $\tau$, $\ocW(\tau)$ is the amount of allowances needed by  the producers who have won electricity market shares 
and want to cover their overall 
 production capacity $\kappa_{j}$.  Obviously we have 
\[ \cW (\tau) \leq \ocW(\tau), ~\forall \tau \in[0,\Pnlt].\] 
Moreover, 
\begin{lemma}\label{lem:Wdecreasing}
The function $\tau \mapsto \cW(\tau)$ is non-increasing: 
\[\cW(t') \leq \cW(t), ~\forall\; 0\leq t < t'\leq \Pnlt.\]
\end{lemma}
The proofs of both Lemma \ref{lem:pelecCroissante} and Lemma \ref{lem:Wdecreasing} can be found in Appendix \ref{appendix:proofLemmas1et2}. 

\subsubsection{Towards an equilibrium strategy}

The main result of the section is the computation of the bounds of the interval 
 in which the coupled carbon market Nash equilibria prices evolve: we demonstrate that there is no possible deviation enabling a Nash equilibrium carbon price outside this interval. The price bounds are elaborated as  specific carbon prices  associated to two explicit strategies,  build from the {\it willing-to-buy-allowances}  functions: the {\it Lower price strategy}, and the  {\it Higher price strategy}.

In order to characterize  further Nash equilibria candidates, evolving in this price interval, we analyze  a third set of strategies that are {\it intermediate strategies}.

Those strategies rely on our last design assumption which prevents the carbon market from market failure:
\[
\cW(0)  \leq \WW: \mbox{ no auction, } \quad \quad 
\ocW(\Pnlt) \geq  \WW :  \mbox{ allowances  shortage.}
\]

\begin{assumption}{\bf On the carbon market design.~}\label{ass:hypoTW}
The available allowances $\WW$ satisfy
\[
\ocW(\Pnlt) < \WW < \cW(0). 
 \]
Moreover,  $\Pnlt$ is chosen such that no producer is sidelined from the game: for all $j$,  $ \tau \mapsto\ocW_j(\tau)$ is not  identically zero on $[0,\Pnlt]$.
\end{assumption}
Assumption \ref{ass:hypoTW} allows to define two  prices of particular interest  for the game analysis:
\begin{align}
& \pguess = \sup \{\tau\in[0,\Pnlt]~\text{s.t.}~\cW(\tau) > \WW\}, \label{eq:lower} \\
\mbox{and } \quad& 
\opguess = \sup \{\tau\in[0,\Pnlt]~\text{s.t.}~\ocW(\tau) > \WW\}.\label{eq:higher}
\end{align}
Observe that we always have $\pguess \leq \opguess$. 

\subsubsection*{Lower price through lower price strategy} 

\begin{lemma} \label{lem:lemmaStratUnder}
Consider any strategy $\bfAunder = (\Aunder_1,\ldots,\Aunder_J)$ such that 
\begin{align}
\Aunder_{j}(\tau) = \left \{
\begin{array}{ll}
\displaystyle \cW_{j}(\pguess), & \mbox{ for }  0 \leq \tau \leq \pguess, \\
\mbox{anything admissible},& \mbox{ for } \tau >  \pguess.
\end{array}
\right.
\end{align}
\item{(i)} $\pcar(\bfAunder)\geq \pguess$.
\item{(ii)} In the case where $\pcar(\bfAunder)=\pguess$,  there is no unilateral favorable deviation that clears the market at a \co2 price lower than $\pguess$.
\end{lemma}
We call the {\it Lower price strategy} $ (\cW_1,\ldots,\cW_J)$, 
 as  $\pcar((\cW_1,\ldots,\cW_J)) = \pguess$ by price definitions \eqref{reacmarche-prix-quotas} and \eqref{eq:lower}.

\begin{proof}
Point {\it (i)} is a consequence of the definition of  
$\pguess = \sup \{ \tau \in [0, \Pnlt], \mbox{ s.t. } \cW(\tau) > \WW\}$.
Since $\Aunder_{j}(\tau) = \cW_{j}(\tau)$ for $\tau \leq \pguess$, it follows that 
$\pcar(\bfAunder) = \sup \{ \tau \in [0, \Pnlt ], \mbox{ s.t. }\sum_{j} \Aunder_{j}(\tau) > \WW\} \geq \pguess$.\medskip

To prove {\it (ii)}, first note that, since we assume $\pcar(\bfAunder) = \pguess$, we have
$\varphi_{j}(\bfAunder) \leq \varphi_{j}(\pguess) =  \frac{1}{e_{j}} \cW_{j}(\pguess)$. Indeed, the carbon market clearing can decrease the global   function $\Off(\bfC^\pguess;\cdot)$ to $\Off(\bfAunder;\cdot)$, but the demand function stay unchanged. So, we still have
$\varphi_{j}(\bfAunder) = \frac{1}{e_{j}} \delta_j(\bfAunder)$. 

Suppose one producer, say Producer~1, deviates and chooses $\widetilde{A}_{1}(\cdot)$ instead of $\Aunder_{1}(\cdot)$. Suppose the new carbon price 
$\widetilde{\tau}:=\pcar(\bfAunder_{-1},\widetilde{A}_{1}) < \pguess$.
Since $\Aunder_{j}(\widetilde{\tau}^{+}) = \Aunder_{j}(\widetilde{\tau})$ for $j \neq 1$, necessarily we have $\widetilde{A}_{1}(\widetilde{\tau}^{+}) \leq  \widetilde{A}_{1}(\widetilde{\tau})$, by definition of $\widetilde{\tau}$. Then $\Delta(\widetilde{A}_{1}) \geq 0$ and it follows that 
$\delta_{1}(\bfAunder_{-1},\widetilde{A}_{1}) \leq  \delta_{1}(\bfAunder)$, but $\delta_{j}(\bfAunder_{-1},\widetilde{A}_{1}) \geq  \delta_{j}(\bfAunder)$ for the others $j\neq 1$. 

If $\pelec(\bfAunder_{-1},\widetilde{A}_{1}) \geq \pelec(\bfAunder)$, the others $j\neq 1$ produce at least electricity for the allowances they have, $\varphi_j(\bfAunder_{-1},\widetilde{A}_{1}) \geq \varphi_j(\bfAunder)$. Since the demand is decreasing we have 
$\varphi_{1}(\bfAunder_{-1},\widetilde{A}_{1}) \leq \varphi_{1}(\bfAunder)$. 

Now, if $\pelec(\bfAunder_{-1},\widetilde{A}_{1}) <  \pelec(\bfAunder)$, the offer of Producer~1 based on his penalized marginal production cost is also greater than $\pelec(\bfAunder_{-1},\widetilde{A}_{1})$. Then $\varphi_{1}(\bfAunder_{-1},\widetilde{A}_{1}) \leq \frac{1}{e_1}\delta_1(\bfAunder_{-1},\widetilde{A}_{1})\leq \varphi_{1}(\bfAunder)$.
\end{proof}

\begin{lemma} \label{lem:lemmaPasNashAtau_guess}
Suppose $\bfA$ is such that $\pcar(\bfA) < \pguess$. Then $\bfA$ is not a Nash equilibrium.
\end{lemma}

\begin{proof}
To prove this lemma we exhibit an unilateral favorable deviation of a producer.

$a)$ Assume first that at least one producer exists, say Producer~1, such that
$\varphi_{1}(\bfA) < \kappa_{1}$ and 
there exists a tax value $\hat{\tau}_{1}$ such that $\pcar(\bfA) < \hat{\tau}_{1} \leq \pguess$ and, $\cW_{1}(\tau) = {e_{1}} \kappa_{1}$ for any $\tau \in [\pcar(\bfA), \hat{\tau}_{1}]$. 

This means that  Producer~1 may sell  $\kappa_{1}$,  for any tax level $\tau$ in $[\pcar(\bfA), \hat{\tau}_{1} ]$, and consequently we have
$c_{1} + \tau e_{1} < \pelec(\tau)$ for $\tau$ in $[\pcar(\bfA),\hat{\tau}_{1} ]$.

Consider a deviation $\widetilde{A}_{1}$ of player 1, such that the resulting clearing price on \co2 market, $\pcar(\bfA_{-1},\widetilde{A}_{1} ) \in [\pcar(\bfA), \hat{\tau}_{1} ]$. 

From Remark \ref{rem:diffDePrix}, we have  
\[C_1 + \tau e_1 \leq \pelec(\pcar(\bfA_{-1},\widetilde{A}_{1} )) \leq \pelec(\bfA_{-1},\widetilde{A}_{1}).\]
This means that Producer~1 may sell  its overall covered capacity:  
$\varphi_{1}(\bfA_{-1},\widetilde{A}_{1} ) = \frac{1}{e_{1}} \delta_{1}(\bfA_{-1},\widetilde{A}_{1} )$. 

Now we define $\tau\mapsto\widetilde{A}_{1}(\tau)$ as follows, for $\varepsilon >0$ arbitrarily small and  $\pcar(\bfA) \leq \tau$, 
\[
\begin{array}{l}
\widetilde{A}_{1}(\pcar(\bfA)) =  e_1\kappa_1 ,\\
\begin{array}{rl}
\widetilde{A}_{1}(\tau) = &  \left(\WW - \sum_{j>1} A_{j}(\tau) - \varepsilon \right)\ind_{\displaystyle\{
\sum_{j\neq 1} A_{j}(\tau) +  \delta_{1}(\bfA)\geq \WW\}} \\
& + e_1\kappa_1 \ind_{\displaystyle\{
\sum_{j\neq 1} A_{j}(\tau) +  \delta_{1}(\bfA)<\WW\}} \quad \mbox{ \bf for } \tau \in (\pcar(\bfA),\hat{\tau}_{1}] \\
= &  A_{1}(\tau), \hspace{4.2cm}   \mbox{ \bf for } \tau >  \hat{\tau}_{1} .
\end{array}
\end{array}
\]

Note that $\widetilde{A}_{1}(\tau) \geq A_{1}(\tau)$ for $\pcar(\bfA)\leq \tau \leq \hat{\tau}_{1}$,  and consequently $\pcar(\bfA_{-1},\widetilde{A}_{1}) \geq  \pcar(\bfA)$.

If $\pcar(\bfA_{-1},\widetilde{A}_{1}) > \pcar(\bfA)$, then $e_1\varphi_{1}(\bfA_{-1},\widetilde{A}_{1})) = \delta(\bfA_{-1},\widetilde{A}_{1}) > \delta(\bfA)\geq e_1\varphi_{1}(\bfA)$, and we get our favorable deviation.

If $\pcar(\bfA_{-1},\widetilde{A}_{1})= \pcar(\bfA)$, we observe that when $\Delta(A_{1}) \geq 0$, we also have $\Delta(\widetilde{A}_{1}) = 0$. 
Then by the \co2 market clearing mechanism, Producer~1 gets ${e_{1}} \kappa_{1}$ allowances instead of $\delta(\bfA)$ and strictly improves its electricity market share. when $\Delta(A_{1}) <  0$, we  have $\widetilde{A}_{1}(\pcar(\bfA)^+)> {A}_{1}(\pcar(\bfA)^+)$, that also insures that Producer~1 increases $\delta(\bfA_{-1},\widetilde{A}_{1}) > \delta(\bfA)$ (see \eqref{reacmarche-quantite-quotas}). 
\medskip

$b)$ Assume now that all producers are either such that 
$\varphi_{j}(\bfA)=\kappa_{j}$ or such that  
$\varphi_{j}(\bfA)< \kappa_{j}$ and  $\cW_{j}(\pcar(\bfA)^+) < {e_{j}} \kappa_{j}$.
 Among the second category, there exists at least 
 one producer (say Producer~1) such that $\varphi_{1}(\bfA) < \varphi_{1}(\pcar(\bfA))$ with $\varphi_{1}(\pcar(\bfA))>0$ (unless to contradict $\pcar(\bfA) < \pguess$). Here we have used the notation \eqref{changeOfNotation} and \eqref{changeOfNotationTau}. 
$\cW_{1}(\pcar(\bfA)^+) < {e_{1}} \kappa_{1}$ means that $c_1 + e_1 \pcar(\bfA) = \pelec(\pcar(\bfA))$ (as $\pelec(\cdot)$ is right-continuous).

A strictly favorable deviation $\widetilde{A}_{1}$ of Producer~1, thus consists in increasing its  ask at the price $\pcar(\bfA)^+$, in order to increase its  $\delta(\bfA_{-1},\widetilde{A}_{1})$ (see \eqref{reacmarche-quantite-quotas}): 
\begin{align*}
\widetilde{A}_{1}(\tau) = &\left(\WW - \sum_{j>1} A_{j}(\tau) - \varepsilon \right)\ind_{\displaystyle\{\pcar(\bfA) <  \tau \}} + e_1\kappa_1  \ind_{\displaystyle\{\pcar(\bfA)=\tau \}}. 
\end{align*}
Then $\pcar(\bfA_{-1},\widetilde{A}_{1}) = \pcar(\bfA)$, $\widetilde{A}_{1}(\pcar(\bfA)) \geq  A_{1}(\pcar(\bfA))$, but $\widetilde{A}_{1}(\pcar(\bfA)^+) > A_{1}(\pcar(\bfA)^+)$, for $\varepsilon$ sufficiently small. This last inequality guarantees that $\delta_{1}(\bfA_{-1},\widetilde{A}_{1}) > \delta_{1}(\bfA)$ and finally 
$\varphi_{1}(\pcar(\bfA)) \geq \varphi_{1}(\bfA_{-1},\widetilde{A}_{1}) > \varphi_{1}(\bfA)$. 
\end{proof}

\subsubsection*{Higher  price through higher price strategy} 

\begin{lemma} \label{lem:lemmaStratUpper}
Consider any strategy $\bfAupper = (\Aupper_{1},\cdots, \Aupper_{J})$ such that 
\begin{align}
\Aupper_{j}(\tau) = \left \{
\begin{array}{ll}
\mbox{anything admissible}, &\mbox{ for } \tau \leq \opguess,\\
\displaystyle\ocW_{j}(\tau), &\mbox{ for }  \tau > \opguess.
\end{array}
\right.
\end{align}
\item{(i)} $\pcar(\bfAupper)\leq \opguess$.
\item{(ii)} There is no unilateral favorable deviation  that clears the market at a \co2 price higher than $\opguess$.
\end{lemma}
We call the {\it Higher price strategy} $ (\ocW_1,\ldots,\ocW_J)$, as  $\pcar((\ocW_1,\ldots,\ocW_J)) = \opguess$ by price definitions \eqref{reacmarche-prix-quotas} and \eqref{eq:higher}.
\begin{proof}
Point {\it(i)} follows directly from the definition of $\opguess$.

To prove  {\it (ii)},  suppose one producer, say Producer~1, chooses its strategy $\widetilde{A}_{1}(\cdot)$ instead of $\Aupper_{1}(\cdot)$, and that the resulting \co2 price is $\widetilde{\tau} := \pcar(\bfAupper_{-1},\widetilde{A}_{1}) > \opguess$. Necessarily, due to the definition of $\bfAupper$, this means that $\ocW_{1}(\widetilde{\tau}) = 0$, which in turn means that
$c_{1} + \widetilde{\tau}e_{1} > \pelec(\widetilde{\tau})$. To conclude, it is sufficient to notice that any Producer $j \neq 1$ obtains what he asks for, i.e.
$\delta_{j}(\bfAupper_{-1},\widetilde{A}_{1})=  \ocW_j(\widetilde{\tau}^{+})$, from which it follows that the 
{\it coupled} electricity  price equals the {\it taxed} electricity price: $\pelec(\bfAupper_{-1},\widetilde{A}_{1}) = \pelec(\widetilde{\tau})$,  and then $\varphi_{1}(\bfAupper_{-1},\widetilde{A}_{1}) = \ocW_i(\widetilde{\tau}) =0$ and the deviation  of 1 is not favorable.
\end{proof} 

A strategy $\bfA$ is said to be {\it effective} if all the producers that bought some allowances produce some electricity: 
\[\forall j, \delta_j(\bfA) >0 \Rightarrow \varphi_j(\pcar(\bfA)) >0.\]

\begin{lemma} \label{lem:lemmaPasNashStrongAtau_barre_guess}
\item{(i)} Let $\bfA$ admissible such that $\pcar(\bfA) >  \opguess$. Then $\bfA$ is not an effective strategy. 
\item{(ii)} Let $\bfA$ admissible such that $\pcar(\bfA) >  \opguess$. Then $\bfA$ is not a strong Nash equilibrium.
\end{lemma}
As a consequence of this lemma, if a producer (or a set of producers) that does not produce electricity,  tries  to block the auction game of the carbon market  by buying all the allowances he can, then there always exists a coalition with favorable  deviation. 

\begin{proof}
$(i).$ Effective means that for all producers  such that $\delta_{j}(\bfA) > 0$, we have $\ocW_{j}(\pcar(\bfA))= e_j \kappa_j $, which is clearly in contradiction with the definition of $\opguess$. 

$(ii).$ Given $\bfA$, such that $\pcar(\bfA) >  \opguess$, we consider the coalition of producers $\mathcal{K}$  such that for $j\in\mathcal{K}$,  $\ocW_{j}(\pcar(\bfA))=0$. $\mathcal{K}$ is clearly non-empty  by definition of $\opguess$.  Consider the following cooperating  deviation of $\mathcal{K}$:  
\begin{align*}
\widetilde{A}_{j}(\cdot) = \Aupper_{j}(\cdot), \quad\mbox{for }j\in \mathcal{K}.
\end{align*}
Then $\pcar(\bfA_{-\mathcal{K}},\widetilde{A}_{\mathcal{K}})< \pcar(\bfA)$, and at least for one member of the coalition $\mathcal{K}$, $\delta_{j}(\bfA_{-\mathcal{K}},\widetilde{A}_{\mathcal{K}})> 0$ when $\ocW_{j}(\pcar(\bfA_{-\mathcal{K}},\widetilde{A}_{\mathcal{K}}))>0$. This means that $\varphi_{j}(\bfA_{-\mathcal{K}},\widetilde{A}_{\mathcal{K}}) >0$ which is a strictly  favorable deviation of $j$, whereas the situation is unchanged for the others  in $\mathcal{K}$ that still produce nothing. 
Thus, we exhibit a coalition that allows a deviation from $\bfA$ that benefits to all of its members, and  that benefits strictly to at least one. Then $\bfA$ is not a strong Nash equilibrium.
\end{proof}

\subsubsection*{Price interval}

From Lemmas  \ref{lem:lemmaPasNashAtau_guess} and   \ref{lem:lemmaPasNashStrongAtau_barre_guess}, we have the following: 
\begin{corollary}
If $\bfA$ is a strong Nash equilibrium,  or if it is an effective  Nash equilibrium, then $\pcar(\bfA) \in [\pguess, \opguess]$.
\end{corollary}

The interval in which the coupled carbon market Nash equilibria prices evolve is then $[\pguess, \opguess]$. This price range is generated by the existing gap between the functions $\cW(\cdot)$ and $\ocW(\cdot)$.  

Thus  a condition for a single unique carbon price is that this gap shrinks to zero: the equality between the two
 {\it willing-to-buy-allowances functionals} occurs i.e. for any value $\tau$, and any producer $i$, the allowances needed to cover the global emissions generated by a player who has won electricity market shares and the allowances needed by  a producer who has won electricity market shares and wants to cover its overall production capacity are the same. Clearly, this is very unlikely to happen. 

\medskip 
It is worth of mentioning that the same lemmas  apply when 
producers have an electricity production power plants portfolio, or when one modifies 
the maximal cap $\mathcal{E}_{j}$ of allowances that  each producer $j$ may buy while one redefines $\tau \mapsto \ocW_j(\tau)$ by 
$$\ocW_j(\tau) = \mathcal{E}_{j} \ind_{\{\varphi_j(\tau) >0\}}.$$

Note that if one increases  the maximal cap,  $\opguess$ increases. 

\subsubsection*{Intermediate strategies} 

Consider any strategy profile $\bfBE = (B_{1}, \cdots, B_{J})$ satisfying the following:
\begin{align}
B_{j}(\tau) = \left \{
\begin{array}{ll}
\displaystyle\cW_{j}(\pguess),  &\mbox{ for } \tau \leq  \pguess, \\
\mbox{ anything admissible}, &\mbox{ for } \pguess < \tau \leq \opguess, \\
\displaystyle \ocW_{j}(\tau), &\mbox{ for } \tau >  \opguess.
\end{array}
\right.
\end{align}
This is not in general an equilibrium, nevertheless we have the following properties:
\begin{lemma}\label{lem:intermediate}
\item{(i)} $\pcar(\bfBE) \in [\pguess,\opguess]$.
\item{(ii)} If there exists a favorable deviation from a producer, say Producer~1, that chooses $\widetilde{B}_{1}$ instead of $B_{1}$,  such that $\pcar(\bfBE_{-1}, \widetilde{B}_{1}) < \pguess$, then there exists another favorable deviation $\widehat{B}_{1}$ defined by 
\[
\widehat{B}_{1} = 
\left\{ 
\begin{array}{ll}
\widetilde{B}_{1}(\tau), &\mbox{ for } \tau > \pguess, \\
\cW_{1}(\pguess), &\mbox{ for } \tau \leq \pguess
\end{array}
\right.
\]
such that $\pcar(\bfBE_{-1}; \widehat{B}_{1}) = \pguess$, and such that 
$\varphi_{1}(\bfBE_{-1},\widehat{B}_{1}) \geq
\varphi_{1}(\bfBE_{-1}, {\widetilde{B}}_{1})$.
\end{lemma}

\begin{proof}
Point {\it(i)} follows directly from Lemma \ref{lem:lemmaStratUnder}-{\it(i)} and Lemma \ref{lem:lemmaStratUpper}-{\it(i)}.

To prove {\it(ii)}, we first observe that, as producers $j\neq 1$ are served first on the carbon market, 
\[
\delta_{1}( \bfB_{-1}; \widetilde{B}_1 ) = \WW - \sum_{j\neq 1} \cW_{j}(\pguess).
\]
Moreover, we have $\pcar(\bfB_{-1}, \widehat{B}_1) = \pguess $, and from the \co2 market mechanism it follows that 
\[
\delta_{1}(\bfB_{-1}, \widehat{B}_1)  \geq \delta_{1}(\bfB_{-1}, \widetilde{B}_1). 
\]
Since $\widetilde{B}_{j}(\pcar(\bfBE_{-1}, \widetilde{B}_{1}))=
\widetilde{B}_{j}(\pcar(\bfBE_{-1}, \widetilde{B}_{1})^{+})= \cW_{j}(\pguess)$ for any $j\neq 1$, it follows that
$\delta_{1}(\bfB_{-1}, \widetilde{B}_1) = \WW - \sum_{j \neq 1} \cW_{j}(\pguess)$.
Indeed, for strategy $(\bfB_{-1}, \widehat{B}_1)$, the producers $j\neq 1$ such that $B_{j}(\pguess^{+}) < \cW_{j}(\pguess)$ receive a quantity of quotas $\delta_{j}(\bfB_{-1}, \widehat{B}_1) \leq \cW_{j}(\pguess)$, from which 
$\delta_{1}(\bfB_{-1}, \widehat{B}_1) = \WW - \sum_{j}\delta_{j}(\bfB_{-1}, \widehat{B}_1) \geq \delta_{1}(\bfB_{-1}, \widehat{B}_1)$. We also deduce that $\varphi_{1}(\bfB_{-1}, \widehat{B}_1) = \frac{1}{e_{1}} \delta_{1}(\bfB_{-1}, \widehat{B}_1)$. 
To conclude, it is sufficient to notice that 
$\varphi_{1}(\bfB_{-1}, \widehat{B}_1) = \frac{1}{e_{1}} \delta_{1}(\bfB_{-1}, \widehat{B}_1) \geq \frac{1}{e_{1}} \delta_{1}(\bfB_{-1}, \widetilde{B}_1) \geq \varphi_{1}(\bfB_{-1}, \widetilde{B}_1)$. 
\end{proof}

The following aims to characterize the form of effective Nash equilibria. 
\begin{corollary}
Let ${\bf E}$ be  an  effective Nash equilibrium (i.e $\pcar(\bfE)\leq \opguess$). Then the following ${\bf E'}$ is also an effective Nash equilibrium: 
\begin{align}
E'_{j}(\tau) = \left \{
\begin{array}{ll}
\cW_{j}(\pguess) ,&\mbox{ for }  \tau \leq  \pguess,\\
E_{j}(\tau), &\mbox{ for }~ \pguess < \tau \leq \opguess, \\
\ocW_{j}(\tau),  &\mbox{ for }  \tau > \opguess.
\end{array}
\right.
\end{align}
\end{corollary}
\begin{proof}
From Lemmas \ref{lem:lemmaStratUnder} and \ref{lem:lemmaStratUpper},  $\pcar(\bfE) \in [\pguess,\opguess]$. 
Consider a deviation that produces a bigger carbon price: Producer~1 deviates from  ${E_1'}$ to $\tilde{E_1'}$  with $\pcar( \tilde{E_1'},\bfEp_{-1})> \opguess $. Then by definition of $ \opguess$, $\varphi_1(\tilde{E_1'},\bfEp_{-1}) = 0$. Indeed, a deviation to this  bigger price is possible only if $\ocW_1(\pcar( (\tilde{E_1'},\bfEp_{-1}))=0$.

Now if Producer~1 deviates from  $E_1'$ to $\tilde{E_1'}$  with $\pcar( (\tilde{E_1'},\bfEp_{-1})) < \pguess $, and if we assume that this deviation is strictly favorable:  $\varphi_1(\tilde{E_1'},\bfEp_{-1}) > \varphi_1(\bfEp)$.  Then according to Lemma \ref{lem:intermediate}, we consider $\hat{E_1'}$ that  gives $\pcar( \hat{E_1},\bfEp_{-1}) = \pguess$.  And we still have that   $\varphi_1(\hat{E_1'},\bfEp_{-1}) > \varphi_1(\bfEp)$. 
But the deviation $(\hat{E_1'},\bfE_{-1})$ from $\bfE$ produces the same price and shares than   $(\hat{E_1'},\bfEp_{-1}))$. Since we also have  $\varphi_1(\bfEp) = \varphi_1(\bfE)$, we get a strictly favorable deviation to $\bfE$ which gives the contradiction. 

Same arguments apply when Producer~1 deviates from  $E_1'$ to $\tilde{E_1'}$  with $\pcar( \tilde{E_1'})$ in $[\pguess,\opguess]$. 
\end{proof}

\section{Conclusion}
Once \co2 is emitted into the atmosphere,  it remains there for more than a century. Estimating its value is an essential indicator for efficiently defining policy. Carbon valuation is crucial for designing markets that foster emission reductions. In this paper, we established the links between an electricity market and a carbon auction market through an analysis of electricity producers' strategies. We proved that they lead to the interval where relevant Nash equilibria evolve,  enabling the computation of equilibrium prices on both markets. 
For each producer, each equilibrium derives the level of electricity produced and the \co2  emissions covered. 

For a given design and set of players, the information provided by the interval may be interpreted as a diagnosis of market behavior in terms of prices and volume. {Indeed, it enables the computation of the \co2 emissions actually released, and opens the discussion of a relevant carbon market in terms of mitigation issues.}

In addition to this analysis of the Nash equilibrium we plan to analyze the electricity production mix, with a particular focus on renewable shares that do not participate in emissions.

\subsubsection*{Acknowledgement}
{This work was partly supported by Grant  0805C0098 from ADEME.}

\appendix
\section{Appendix}

\subsection{Proof of Proposition \ref{propo-Nash}} \label{appendix:proofNash}

\noindent
{\bf A. First we prove the dominance property \textit{(i)}.} 
\medskip

Suppose that one producer, let us say producer $1$, deviates and chooses  $C_{1}$ 
instead of $s_{1}$. We have to show that its market share cannot be reduced by this deviation. 
By definition of the admissibility (see \eqref{eq:classeStratAdmiss}) we have
\[
s_{1}(q) \geq C_{1}(q), \forall q\in [0,\kappa_1].
\]
Hence the offer functions defined by \eqref{defOffrej} satisfy 
${\Of}(s_{1};\cdot) \leq {\Of}(C_{1};\cdot)$. 
By adding the unchanged offers of the other producers 
\begin{equation}\label{in1}
\Off((\bfs_{-1}, s_{1});\cdot) \leq \Off((\bfs_{-1}, C_{1});\cdot), 
\end{equation}
where $(\bfs_{-1}, C_{1})$ denotes the strategy profile that includes Producer~1 deviation.
The minimum market clearing price  \eqref{reacmarche-prix}  for strategy profile  $\bfs$ is
\begin{align*}
\uup(\bfs) = \inf \{ p , \;  \Off(\bfs ; p) >  D(p) \}. 
\end{align*}
The minimum  market clearing price \eqref{reacmarche-prix} for strategy profile $(\bfs_{-1},C_{1})$ 
is
\begin{align*}
\uup(\bfs_{-1},C_{1}) =  \inf \{ p , \;  \Off((\bfs_{-1},C_{1}) ; p) >  D(p) \}. 
\end{align*}
The inequality \eqref{in1} together with the fact that the demand $D(\cdot)$ is a non-increasing  
function imply that $\uup(\bfs_{-1},C_{1}) \leq \uup(\bfs)$,
from which,  with \eqref{regleChoixPrix} we deduce that
\[
\pelec(\bfs_{-1},C_{1})\leq \pelec(\bfs). 
\]

Now let us show that Producer~1 does not reduce its market share by deviating from   $s_{1}(\cdot)$ to $C_{1}(\cdot)$, that is $\varphi_{1}(\bfs_{-1}, C_{1})  \geq  \varphi_{1}(\bfs)$. 

\medskip
For the sake of clarity we adopt, in this paragraph, the following notation:
\begin{align*}
\begin{array}{ll}
\uup_{\bfs}  &:=  \uup(\bfs)  \\
\pelec_{\bfs} & :=  \pelec(\bfs) 
\end{array}
\quad \mbox{ and } \quad 
\begin{array}{ll}
\uup_{\bfs C} &:= \uup(\bfs_{-1},C_{1})  \\
\pelec_{\bfs C}& := \pelec((\bfs_{-1},C_{1})). 
\end{array}
\end{align*}

\medskip\noindent
{\bf We first consider the case where $\uup_{\bfs C}< \uup_{\bfs}$. ~} ~By definition of the minimum clearing price $\uup_{\bfs C}$, the fact that $D(\uup_{\bfs})  \leq D(\uup_{\bfs C})$ and the fact that $\Off((\bfs_{-1},C_{1});\cdot)$ is non-decreasing, we have 
\[
D(\pelec_{\bfs})\leq D(\uup_{\bfs})  \leq D(\uup_{\bfs C}) \leq \Off((\bfs_{-1},C_{1});\uup_{\bfs C}) \leq 
\Off((\bfs_{-1},C_{1});\pelec_{\bfs C}). 
\]
Hence,  
\begin{align*}
\Off((\bfs_{-1}, s_{1}), \pelec_{\bfs})\wedge D(\pelec_{\bfs}) & \leq \Off((\bfs_{-1},C_{1});\pelec_{\bfs C})\wedge D(\pelec_{\bfs C}),\\
\Off((\bfs_{-1}, s_{1}), \uup_{\bfs})\wedge D(\uup_{\bfs}) & \leq \Off((\bfs_{-1},C_{1});\uup_{\bfs C})\wedge D(\uup_{\bfs C}).
\end{align*}
From the market clearing \eqref{eq:clearingQuantities}	 we get
\begin{align*}
\varphi_{1}(\bfs_{-1}, s_{1}) -  \varphi_{1}(\bfs_{-1},C_{1})  =& 
\Off((\bfs_{-1}, s_{1}), \pelec_{\bfs})\wedge D(\pelec_{\bfs}) - \Off((\bfs_{-1},C_{1});\pelec_{\bfs C})\wedge D(\pelec_{\bfs C})\\&
+ \sum_{j>1} \left(\varphi_j(\bfs_{-1},C_{1}) - \varphi_j(\bfs_{-1}, s_{1})\right). 
\end{align*}
According to Definition \ref{def:quantities}, let us denote
\begin{align*}
\mathcal{E}(\uup_{\bfs}) = \left\{j \in \{2,\ldots,J\}\mbox{ s.t.  }\Delta^- \Of(s_j ; \uup_{\bfs}) > 0 \right\}.
\end{align*}
We have 
\begin{align*}
\varphi_{1}(\bfs_{-1}, s_{1}) -  \varphi_{1}(\bfs_{-1},C_{1}) =& 
\Off((\bfs_{-1}, s_{1}); \pelec_{\bfs})\wedge D(\pelec_{\bfs}) - \Off((\bfs_{-1},C_{1});\pelec_{\bfs C})\wedge D(\pelec_{\bfs C}) \\&
+ \sum_{j >  1, j\notin \mathcal{E}(\uup_{\bfs})} \left(\varphi_j(\bfs_{-1},C_{1}) - \Of(s_j; \pelec_{\bfs})\right) \\&
 +   \sum_{j > 1, j\in \mathcal{E}(\uup_{\bfs})}  (\varphi_j(\bfs_{-1},C_{1})  - \varphi_j(\bfs_{-1}, s_{1}))\\
\leq & 
\Off((\bfs_{-1}, s_{1}); \pelec_{\bfs})\wedge D(\pelec_{\bfs}) - \Off((\bfs_{-1},C_{1});\pelec_{\bfs C})\wedge D(\pelec_{\bfs C}) \\&
+ \sum_{j >  1, j\notin \mathcal{E}(\uup_{\bfs})} \left( \Of(s_j;\pelec_{\bfs C}) - \Of(s_j; \pelec_{\bfs} )\right) \\&
 +   \sum_{j > 1, j\in \mathcal{E}(\uup_{\bfs})}  (\varphi_j(\bfs_{-1},C_{1})  - \varphi_j(\bfs_{-1}, s_{1})).
\end{align*}
Since $\pelec_{\bfs C} \leq \pelec_{\bfs}$  we get
\begin{align*}
\varphi_{1}(\bfs_{-1}, s_{1}) -  \varphi_{1}(\bfs_{-1},C_{1}) \leq &    
\sum_{j > 1, j\in \mathcal{E}(\uup_{\bfs})}(\varphi_j(\bfs_{-1},C_{1})  - \varphi_j(\bfs_{-1}, s_{1})).
\end{align*}
But for any $j\in \mathcal{E}(\uup_{\bfs})$, the quantity  
$\Of(s_{j}; \uup_{\bfs}^{-}) \leq \varphi_j(\bfs_{-1}, s_{1})$.  As $\Of(\bfs_{j} ; \cdot)$ is non-decreasing ans since we have assumed $\uup_{\bfs C} < \uup_{\bfs}$,  we get
$$
\Of(s_j; \uup_{\bfs C}^{-}) \leq \Of(s_j; \uup_{\bfs}^-) \leq \varphi_j(\bfs_{-1}, s_{1}).
$$
For such  $j>1$ we thus have
$$
\varphi_j(\bfs_{-1},C_{1})  - \varphi_j(\bfs_{-1},s_{1}) 
\leq 
\varphi_j(\bfs_{-1},C_{1})- \Of(s_j; {\pelec_{\bfs}}^-) \leq  \varphi_j(\bfs_{-1},C_{1})- \Of(s_j;\pelec_{\bfs C}) \leq 0, $$
from which it follows that 
$
\varphi_{1}(\bfs_{-1}, s_{1}) -  \varphi_{1}(\bfs_{-1},C_{1}) \leq 0.
$

\medskip\noindent
{\bf Now consider the case where  $\uup_{\bfs} = \uup_{\bfs C}:= \uup$. ~}~
Due to the market rule (\ref{regleChoixPrix}), we necessarily  have $\pelec_{\bfs} = \pelec_{\bfs C} :=\pelec$.

\medskip\noindent
{$\bullet$ If  $\Off((\bfs_{-1}, s_{1});\pelec) \leq \Off((\bfs_{-1}, C_{1});\pelec)\leq D(\pelec)$, } then 
by the market clearing 
\begin{align*}
\varphi_{1}(\bfs_{-1}, s_{1})   = \Of(s_1;\pelec) \leq \Of(C_1;\pelec) = \varphi_{1}(\bfs_{-1},C_{1}).
\end{align*} 

\medskip\noindent
{$\bullet$ If  $\Off((\bfs_{-1}, s_{1});\pelec) \leq D(\pelec)\leq \Off((\bfs_{-1}, C_{1});\pelec)$, }then
\begin{align*}
\varphi_{1}(\bfs_{-1}, s_{1})   = \Of(s_1;\pelec) \leq & D(\pelec) - \sum_{j>1} \varphi_{j}(\bfs_{-1},s_{1}) = D(\pelec) - \sum_{j>1} \Of(s_j;\pelec)  \\
\leq  & D(\pelec) - \sum_{j>1} \varphi_j(\bfs_{-1},C_{1})  = \varphi_1(\bfs_{-1},C_{1}).
\end{align*}

\medskip\noindent
{$\bullet$ If $D(\pelec) < \Off((\bfs_{-1}, s_{1});\pelec) \leq \Off((\bfs_{-1},C_{1});\pelec)$, }
by the market clearing we get
\begin{align*}
\varphi_{1}(\bfs_{-1}, s_{1}) -  \varphi_{1}(\bfs_{-1},C_{1})  =& \Off((\bfs_{-1}, s_{1}), \pelec)\wedge D(\pelec)  - \Off((\bfs_{-1},C_{1});\pelec)\wedge D(\pelec) \\
& + \sum_{j>1} (\varphi_j(\bfs_{-1},C_{1})  - \varphi_j(\bfs_{-1}, s_{1}))\\
\leq& \sum_{j>1} (\varphi_j(\bfs_{-1},C_{1})  - \varphi_j(\bfs_{-1}, s_{1}))\\
\leq& \sum_{j>1,j\in\mathcal{E}(\uup)} (\varphi_j(\bfs_{-1},C_{1})  - \varphi_j(\bfs_{-1}, s_{1})).
\end{align*}
From \eqref{reacmarche-quantite-elec}, we have for $j \in \mathcal{E}(\uup)$
\begin{align*}
& \varphi_j(\bfs_{-1},s_{1}) =  \Of(s_j,\uup^{-})  + \Delta^- \Of(s_j;\uup) \dfrac{\left( D(\uup) - \Off((\bfs_{-1}, s_{1}),\uup^{-})\right)}{\displaystyle \Delta^- \Off((\bfs_{-1}, s_{1}),\uup)}\\
\mbox{and } \quad \mbox{ }  & \varphi_j(\bfs_{-1},C_{1})  =  \Of(s_j;\uup^{-})  + \Delta^- \Of(s_j;\uup) \dfrac{\left( D(\uup) - \Off((\bfs_{-1},C_{1});\uup^{-})\right)}{\displaystyle \Delta^- \Off((\bfs_{-1},C_{1});\uup)}.
\end{align*}
Hence, if $\mathcal{E}(\uup)$ is non empty then at least one producer exists, $j \neq 1$ such that
$\Delta^-\Of(s_j;\uup) > 0$. and from the desegregation of  $\Off$ and definition of  $\Delta^-$ it results that
\begin{align*}
&\varphi_{1}(\bfs_{-1}, s_{1}) -  \varphi_{1}(\bfs_{-1}, C_{1})\\
& = \sum_{j>1,j\in\mathcal{E}(\uup)} \Delta^- \Of(s_j,\uup) 
\left( 
\dfrac{\left( D(\uup) - \Off(\bfs_{-1},\uup^{-}) - \Of(C_1;\uup^{-})\right)}
      {\displaystyle \Off((\bfs_{-1}, C_{1});\uup)-\Off(\bfs_{-1},\uup^{-}) -\Of(C_1;\uup^{-}) } \right. \\
&\hspace{5cm}\left. 
- \dfrac{\left( D(\uup) - \Off(\bfs_{-1},\uup^{-}) - \Of(s_1,\uup^{-}) \right) }
    {\displaystyle \Off((\bfs_{-1}, s_{1}),\uup)- \Off(\bfs_{-1},\uup^{-}) - \Of(s_1;\uup^{-}) }
\right). 
\end{align*}
We note that 
\begin{align*}
0 & < \Off((\bfs_{-1}, s_{1});\uup)- \Off(\bfs_{-1},\uup^{-}) - \Of(C_1;\uup^{-}) \\
& <   \Off((\bfs_{-1}, C_{1});\uup)-\Off(\bfs_{-1},\uup^{-}) -\Of(C_1;\uup^{-}), 
\end{align*}
and that  $D(\uup) - \Off((\bfs_{-1}, C_{1});\uup^{-}) >0$ by definition of $\uup$. Then
\begin{align*}
&\varphi_{1}(\bfs_{-1}, s_{1}) -  \varphi_{1}(\bfs_{-1}, C_{1})\\
&\quad \leq \sum_{j>1,j\in\mathcal{E}(\uup)}\Delta^- \Of(s_j;\uup)
 \times  
 \left( 
 \dfrac{\left( D(\uup) - \Off(\bfs_{-1},\uup_-) - \Of(C_1;\uup_-)\right)}
       {\displaystyle  \Off((\bfs_{-1}, s_{1});\uup)- \Off(\bfs_{-1},\uup^{-}) - \Of(C_1;\uup^{-}) } 
 \right. \\
&\hspace{5cm}\left.       - 
 \dfrac{\left( D(\uup) - \Off(\bfs_{-1},\uup_-) - \Of(s_1;\uup_-)\right)}
       {\displaystyle \Off((\bfs_{-1}, s_{1});\uup)- \Off(\bfs_{-1},\uup^{-}) - \Of(s_1;\uup^{-}) }
 \right). 
\end{align*}
Since $D(\uup) \leq \Off((\bfs_{-1}, s_{1});\uup)$ and $\Of(C_1; \uup^{-})\geq \Of(s_1 ; \uup^{-})$, we can deduce 
that
$
\varphi_{1}(\bfs_{-1}, s_{1}) -  \varphi_{1}(\bfs_{-1}, C_{1}) \leq 0.
$
This follows from the fact that when $A \leq B$, the map   $x \mapsto\dfrac{A-x}{B-x}$  is decreasing on $[0,A)$.

\medskip\noindent
{\bf B. We prove the  uniqueness property \textit{(iii)}:} all Nash equilibria induce the same  electricity price and same quantities of electricity bought to each producer.  

First, we state the following  consequence of the dominance property  \textit{(i)}: 
\begin{lemma}\label{lem:dominancebis}
For any admissible strategy ${\bf{s}} = (s_1,\ldots,s_J)$, such that $\uup(\bfs) = \uup(\bfC)$, if  producer $j$ is such that  $s_j = C_j$,  then 
$$ \varphi_j(\bfs) \geq  \varphi_j(\bfC).$$
\end{lemma}
\begin{proof}
As arguments are very similar to the proof of \textit{(i)}, we just sketch them. 
Let $\bfs$ such that $\uup(\bfs)=\uup(\bfC):= \uup $. Assume that Producer~1 is such that  $s_1 =C_1$.\\
{$\bullet$ If $\Off(\bfs;\uup) \leq D(\uup)$, } then by the market clearing 
\begin{align*}
\varphi_{1}(\bfs) = \Of(s_1;\uup)  = \Of(C_1;\uup) \geq  \varphi_{1}(\bfC). 
\end{align*}
{$\bullet$ If $D(\uup) <  \Off(\bfs;\uup) \leq \Off(\bfC;\uup) $, } by the market clearing we get 
\begin{align*}
& \varphi_1(\bfs_{-1},C_{1})  =  \Of(C_1;\uup^-)  + \Delta^- \Of(C_1;\uup) \dfrac{\left( D(\uup) - \Off((\bfs_{-1};C_{1});\uup^-)\right)}{\displaystyle \Delta^- \Off((\bfs_{-1};C_{1});\uup)} \\
\mbox{and } \quad \mbox{ } &
\varphi_1(\bfC_{-1},C_{1}) =  \Of(C_1;\uup^-)  + \Delta^- \Of(C_1;\uup) \dfrac{\left( D(\uup) - \Off((\bfC_{-1}; C_{1}),\uup^-)\right)}{\displaystyle \Delta^- \Off((\bfC_{-1}; C_{1}),\uup)}. 
\end{align*}
Thus, 
\begin{align*}
&\varphi_1(\bfs_{-1},C_{1}) - \varphi_1(\bfC_{-1},C_{1})\\
& = 
\Delta^- \Of(C_1;\uup) \left( 
\dfrac{\left( D(\uup) - \Off((\bfs_{-1},C_{1});\uup^-)\right)}{\displaystyle \Off((\bfs_{-1}, C_{1});\uup)- \Off((\bfs_{-1},C_{1});\uup^-)} 
- \dfrac{\left( D(\uup) - \Off((C_{-1},C_{1});\uup^-)\right)}{\displaystyle \Off((C_{-1}, C_{1});\uup)- \Off((C_{-1},C_{1});\uup^-)}
\right).
\end{align*}
Assuming that $\Delta^- \Of(C_1;\uup)>0$, we note that  
\begin{align*}
0 < {\displaystyle \Off((\bfs_{-1}, C_{1});\uup)- \Off((\bfs_{-1},C_{1});\uup^-)}  \leq  {\displaystyle \Off((C_{-1}, C_{1});\uup)- \Off((\bfs_{-1},C_{1});\uup^-)} .
\end{align*}
Since $D(\uup) - \Off((C_{-1}, C_{1});\uup^-) >0$ by definition of $\uup$,  
\begin{align*}
&\varphi_1(\bfs_{-1},C_{1}) - \varphi_1(C_{-1},C_{1})\\
&\geq \Delta^- \Of(C_1;\uup)
\left( 
\dfrac{\left( D(\uup) - \Off((\bfs_{-1},C_{1});\uup^-)\right)}{\displaystyle \Off((C_{-1}, C_{1});\uup)- \Off((\bfs_{-1},C_{1});\uup^-)}  
- \dfrac{\left( D(\uup) - \Off((C_{-1},C_{1});\uup^-)\right)}{\displaystyle \Off((C_{-1}, C_{1});\uup)- \Off((C_{-1},C_{1});\uup^-)}\right).
\end{align*}
As $\Off((\bfs_{-1},C_{1});\uup^-) \leq \Off((C_{-1},C_{1});\uup^-)$,  we get 
$\varphi_1(\bfs_{-1},C_{1}) - \varphi_1(C_{-1},C_{1}) \geq 0.$
\end{proof}

\medskip\noindent
{\bf We prove that the quantities are the same for all Nash equilibria.  }
Let $\bfw$ an other Nash equilibrium that differs from $\bfC$.  On the global offers we always have 
$
\Off(\bfw;\cdot) \leq \Off(\bfC;\cdot)
$
that implies
$
 \uup(\bfw) \geq \uup(\bfC).
$
Note that  when $\uup(\bfC) = \plol$, all admissible strategies $\bfs$ are Nash as 
$\varphi_j(\bfC) =  \varphi_j(\bfs) = \kappa_j,$ for all $j.$

By the offers ordering, it is straightforward to show that 
\[\sum_{j=1}^J \varphi_j(\bfw) \leq  \sum_{j=1}^J \varphi_j(\bfC).\]
Assume that the quantities are not the same, then there exists a producer, say Producer~1, such that $\varphi_1(\bfw) < \varphi_1(\bfC).$
And we also have 
\begin{align*}
\varphi_1(\bfw) < \varphi_1(\bfC) \leq \Of(C_1; \pelec(\bfC)) \leq \Of(C_1;\pelec(\bfw_{-1},C_1). 
\end{align*}

If $\uup(\bfC) = \uup(\bfw_{-1},C_1)$, then by Lemma \ref{lem:dominancebis}, we have that $ \varphi_1(\bfw_{-1},C_1) \geq \varphi_1(\bfC)$ and hence  $ \varphi_1(\bfw_{-1},C_1) > \varphi_1(\bfw)$. In other words, $\bfw$ has a strictly favorable deviation for Producer~1 that contradicts the assumption that $\bfw$ is a Nash equilibrium.

Now if $\uup(\bfC) < \uup(\bfw_{-1},C_1)$, by \eqref{eq:clearingQuantity}, 
\begin{align*}
\varphi_1(\bfw) < \varphi_1(\bfC) \leq \Of(C_1; \uup(\bfC)) \leq 
\Of(C_1;\uup((\bfw_{-1},C_1))^-)\leq \varphi_1(\bfw_{-1},C_1),
\end{align*}
and the same conclusion follows.

\medskip\noindent
{\bf We prove that the equilibrium best bid price is unique: $\op(\bfw) = \op(\bfC)$, for  an other Nash equilibrium $\bfw$.} 
Assume the contrary,  $\op(\bfw) >  \op(\bfC)$. Then by the definition of $\op(\cdot)$, we have that 
$
D(\uup(\bfw)) < D(\uup(\bfC)).
$

From \eqref{eq:clearingQuantities} and  \eqref{eq:clearingQuantity}, 
\[\sum_{j=1}^J \varphi_j(\bfw) \leq  D(\uup(\bfw))  < D(\uup(\bfC)^+) \leq D(\uup(\bfC)) \wedge \Off(\bfC;\uup(\bfC)) = \sum_{j=1}^J \varphi_j(\bfC) \]
that contradicts the fact that Nash equilibria have same clearing quantities.

\subsection{Proofs of Lemma \ref{lem:pelecCroissante} and Lemma \ref{lem:Wdecreasing}}
\label{appendix:proofLemmas1et2}
\newcommand{\SD}{S_{D}}
\newcommand{\SK}{S_{\kappa}}

\subsubsection*{Proof of Lemma \ref{lem:pelecCroissante}.} 

Although the result of this lemma is  intuitive, the proof is rather  technical. This is due to  our assumptions,  in particular regarding demand, that  allow  the demand function to have discontinuity points and some non-elasticity areas (see Assumption \ref{ass:demande}). 

More precisely, if we define the map $\tau \mapsto \Off(\tau; p)$ by 
\[\Off(\tau; p) = \sum_{i=1}^J \Of(C_{j}^{\tau}(\cdot);p) = \sum_{i=1}^J \kappa_i \ind_{\{p \geq c_i + 
\tau e_i\}} = \sum_{i=1}^J \kappa_i \ind_{\{\tau \leq  \frac{ p - c_i}{e_i}\}},  \]
then we can observe that, for any $p>0$ far enough from the $c_i$, and any $\tau' \geq \tau$,
\[
\Off(\tau'; p) \leq \Off(\tau; p) \quad 
\mbox{ and } \quad 
\lim_{\epsilon \rightarrow 0^{+}}\Off(\tau+\epsilon; p) = \Off(\tau; p).
\]

We call $\SD = \{p_d ; \lim_{\epsilon \rightarrow 0^+} D(p_d + \epsilon) < D(p_d)\}$, the set of discontinuity points of the   
Demand function. 

We call $\SK = \{p_c ; D(p_c) = \sum \kappa_i\}$, the set of prices that make demand coincide with some accumulation of  production capacities. 

We observe that $\pelec(\tau) \in  \{ c_i + \tau e_i , i=1,\ldots,j\} \cup \SD\cup \SK$.  In particular, from Definition \ref{def:clearingElec}, 
$\uup(\tau) = \inf\{p>0; \Off(\tau;p) > D(p)\},$
and we obtain that
$D(\uup(\tau+\epsilon)) \leq \Off(\tau+\epsilon; \uup(\tau+\epsilon))\leq \Off(\tau; \uup(\tau+\epsilon))$
from which we conclude that $\uup(\tau+\epsilon) \geq \uup(\tau)$. 

Now we prove the right continuity of $\tau \mapsto \uup(\tau)$. Let us fix a $\tau$.

\subsubsection*{(i) We first consider the case $D(\uup(\tau)) < \Off(\tau; \uup(\tau))$. }This means that $\uup(\tau)$ is of the form $c_\ell + \tau \ell$, for a given $\ell$. Then when $\epsilon >0$ is small enough, we also have $\uup(\tau + \epsilon) = c_\ell + (\tau + \epsilon) e_\ell$. Indeed, 
$D(c_\ell + (\tau + \epsilon) e_\ell) \leq D(c_\ell + \tau e_\ell)$ and for a small enough $\epsilon$, 
\[
\Off(\tau;c_\ell + \tau e_\ell) = \kappa_\ell + 
\sum_{i\neq\ell}\kappa_i \ind_{\{\tau \leq \frac{c_\ell - c_i }{1 - e_i/e_\ell}}\} = \Off(\tau +\epsilon;c_\ell + (\tau + \epsilon) e_\ell).
\]
Thus, $D(c_\ell + (\tau + \epsilon) e_\ell) < \Off(\tau +\epsilon;c_\ell + (\tau + \epsilon) e_\ell)$ which implies that $\uup(\tau) + e_\ell\epsilon = c_\ell + (\tau + \epsilon)e_\ell\geq \uup(\tau+\epsilon)$ and hence
$ e_\ell\epsilon \geq \uup(\tau+\epsilon) - \uup(\tau).$

\subsubsection*{(ii) We consider next  the case $D(\uup(\tau)) >   \Off(\tau; \uup(\tau))$. }This means that $\uup(\tau)\in \SD$ is at a discontinuity point, say $p_d$ of the demand,  $\uup(\tau)= p_d$. 
Then, for any $\delta>0$,
\[D(\uup(\tau) + \delta) < \Off(\tau; \uup(\tau) + \delta).\] 
But
\[
\Off(\tau; p_d +\delta) = \sum_{i=1}^J \kappa_i \ind_{\{\tau \leq  \frac{ p_d + \delta  - c_i}{e_i}\}}
\]
and we can choose $\delta$ to be small enough so that $\tau \neq \frac{ p_d + \delta  - c_i}{e_i}$. 
Then, for a small enough $\epsilon$, 
\[
D(\uup(\tau) +\delta) < \Off(\tau; \uup(\tau) + \delta) = \Off(\tau + \epsilon; \uup(\tau) + \delta),
\] 
which implies that $\uup(\tau) + \delta \geq \op(\tau+\epsilon)$, so we obtain
$ \delta \geq \uup(\tau+\epsilon) - \uup(\tau) \geq 0.$

\subsubsection*{(iii) We consider now the case $D(\uup(\tau)) =   \Off(\tau; \uup(\tau))$. }This means that $\uup(\tau)\in \SK$, say $\uup(\tau)=p_c$  
Then, for any $\delta>0$,
\[D(\uup(\tau) + \delta) < \Off(\tau; \uup(\tau) + \delta).\] 
But, 
\[\Off(\tau; p_c +\delta) = \sum_{i=1}^J \kappa_i \ind_{\{\tau \leq  \frac{ p_c + \delta  - c_i}{e_i}\}}\]
and we can choose $\delta$ small enough such that $\tau \neq \frac{ p_c + \delta  - c_i}{e_i}$. 
Then, for $\epsilon$ small enough, 
\[D(\uup(\tau) +\delta) < \Off(\tau; \uup(\tau) + \delta) = \Off(\tau + \epsilon; \uup(\tau) + \delta)\] 
which implies that $\uup(\tau) + \delta \geq \uup(\tau+\epsilon)$, so we get 
$\delta \geq \uup(\tau+\epsilon) - \uup(\tau) \geq 0.$
The right-continuity of $\tau \mapsto \op(\tau)$ follows, by definition as $\op(\tau)$ is a continuous transformation of $\uup(\tau)$.

\subsubsection{Proof of Lemma \ref{lem:Wdecreasing}.} 

The proof consists in a complete analysis of the entire combination of situations, but each situation is elementary. 

Let us suppose the opposite, that is there exists $0\leq t < t'\leq \Pnlt$ such that the emission levels are 
$\cW(t') >  \cW(t)$.

We define the function $\tau \mapsto I(\tau)$ valued in the subsets of $\{1,\ldots,J\}$ that lists the producers in the electricity  market producing at tax level $\tau$: 
\[i\in I(\tau)\quad\mbox{if} \quad  \varphi_i(\tau) > 0.\]
In particular we have for all $\tau\in[0,\Pnlt]$, 
\[\cW(\tau) = \sum_{i\in I(\tau)} e_i \varphi_i(\tau).\] 
\subsubsection*{(i) We first examine the situation  $I(t') = I(t)$.}

To shorten the expressions, we adopt the following shortened notation
\[ I(t) = I \quad \mbox{ and  } \quad I(t') = I'.\]

{\bf (i-a)} If $\sum_{i\in I}\varphi_i(t) = D(t)$ then, from the demand constraint (DC)  and the emission levels hypothesis (EH), we have 

\renewcommand{\theequation}{DC}
\begin{align}
 \sum_{i\in I}\varphi_i(t) = D(t) \geq  D(t')\geq \sum_{i\in I'}\varphi_i(t')
\end{align}
\renewcommand{\theequation}{EH}
\begin{align}
 \sum_{i\in I}\varphi_i(t) e_i  <  \sum_{i\in I'}\varphi_i(t')e_i.
\end{align}
\renewcommand{\theequation}{\thesection.\arabic{equation}}
We denote by $\widehat{I}$ the subset of $I$ of index   such that $c_i+t e_i = \uup(t)$. In particular, when $j\in I\setminus\widehat{I}$, then $\varphi_j(t) = \kappa_j$. 

Note that there exists at most one index (say $\ell$) in the set $\widehat{I}\cap\widehat{I'}$. 
If $j \in \widehat{I}\setminus\widehat{I'}$ and $k\in\widehat{I'}\setminus\widehat{I}$,  then, by the definition of the sets
\begin{align*}
\begin{array}{ll}
c_j + e_j t = c_{\ell} + e_{\ell} t, & \quad c_k + e_k t < c_j + e_j t,\\
c_j + e_j t' < c_{\ell} + e_{\ell} t', & \quad  c_k + e_k t' = c_{\ell} + e_{\ell} t',\\
 c_j + e_j t' < c_k + e_k t , & \quad c_k + e_k t < c_{\ell} + e_{\ell} t,
\end{array}
\end{align*}
from which, we easily deduce that 
\begin{align}\label{eq:classement_rate}
\max\{e_j, j\in \widehat{I}\setminus\widehat{I'}\} <  e_{\ell} < \min\{e_k,k\in \widehat{I'}\setminus\widehat{I}\}.
\end{align}
 
Now we decompose the sets $I$ and $I'$ in  the demand constraint (DC) and the emission levels hypothesis (EH) as follows: 
\renewcommand{\theequation}{DC}
\begin{align}
& \sum_{n\in I\setminus\widehat{I}\cup \widehat{I'}}\kappa_n + \varphi_{\ell}(t) + \sum_{i\in \widehat{I}\setminus\widehat{I'}}\varphi_i(t) 
+ \sum_{k\in \widehat{I'}\setminus\widehat{I}}\kappa_k 
 \geq  \sum_{n\in I\setminus\widehat{I}\cup \widehat{I'}}\kappa_n + \varphi_{\ell}(t') + \sum_{i\in \widehat{I}\setminus\widehat{I'}}\kappa_i
+ \sum_{k\in \widehat{I'}\setminus\widehat{I}}\varphi_k(t'),
\end{align}
\renewcommand{\theequation}{EH}
\begin{align}
& \sum_{n\in I\setminus\widehat{I}\cup \widehat{I'}}e_n \kappa_n + e_{\ell} \varphi_{\ell}(t) + \sum_{i\in \widehat{I}\setminus\widehat{I'}}e_i \varphi_i(t) 
+ \sum_{k\in \widehat{I'}\setminus\widehat{I}}e_k\kappa_k  \nonumber   \\
& \qquad< \sum_{n\in I\setminus\widehat{I}\cup \widehat{I'}}e_n \kappa_n + e_{\ell} \varphi_{\ell}(t') + \sum_{i\in \widehat{I}\setminus\widehat{I'}}e_i \kappa_i
+ \sum_{k\in \widehat{I'}\setminus\widehat{I}}e_k \varphi_k(t').
\end{align}
After simplification, we obtain
\renewcommand{\theequation}{DC}
\begin{align}
&  \varphi_{\ell}(t) + \sum_{i\in \widehat{I}\setminus\widehat{I'}}\varphi_i(t) 
+ \sum_{k\in \widehat{I'}\setminus\widehat{I}}\kappa_k \geq  \varphi_{\ell}(t') + \sum_{i\in \widehat{I}\setminus\widehat{I'}}\kappa_i
+ \sum_{k\in \widehat{I'}\setminus\widehat{I}}\varphi_k(t'), 
\end{align}
\renewcommand{\theequation}{EH}
\begin{align}
&  e_{\ell} \varphi_{\ell}(t) + \sum_{{\small i\in \widehat{I}\setminus\widehat{I'}}}e_i \varphi_i(t) 
+ \sum_{k\in \widehat{I'}\setminus\widehat{I}}e_k\kappa_k< e_{\ell} \varphi_{\ell}(t') + \sum_{i\in \widehat{I}\setminus\widehat{I'}}e_i \kappa_i
+ \sum_{k\in \widehat{I'}\setminus\widehat{I}}e_k \varphi_k(t').
\end{align}
\renewcommand{\theequation}{\thesection.\arabic{equation}}
Assume first that $\varphi_{\ell}(t) + \sum_{i\in \widehat{I}\setminus\widehat{I'}}\varphi_i(t)\geq \varphi_{\ell}(t') + \sum_{i\in \widehat{I}\setminus\widehat{I'}}\kappa_i$. Equivalently,  we have 
\[\varphi_{\ell}(t)- \varphi_{\ell}(t') \geq  \sum_{i\in \widehat{I}\setminus\widehat{I'}}(\kappa_i -\varphi_i(t))\]
and from \eqref{eq:classement_rate}, 
\[e_{\ell}\left(\varphi_{\ell}(t)- \varphi_{\ell}(t')\right) \geq \sum_{i\in \widehat{I}\setminus\widehat{I'}}e_i (\kappa_i -\varphi_i(t)).\]
By combining the above with the emission levels hypothesis (EH), we obtain the following contradiction: $\sum_{k\in \widehat{I'}\setminus\widehat{I}}e_k\kappa_k < \sum_{k\in \widehat{I'}\setminus\widehat{I}}e_k \varphi_k(t')$. 

Assume now that $\varphi_{\ell}(t) + \sum_{i\in \widehat{I}\setminus\widehat{I'}}\varphi_i(t) <  \varphi_{\ell}(t') + \sum_{i\in \widehat{I}\setminus\widehat{I'}}\kappa_i$.  Multiplying the demand constraint (DC) by  $\hat{e}:=\min\{e_k,k\in \widehat{I'}\setminus\widehat{I}\}$, we get  
\begin{align*}
\sum_{k\in \widehat{I'}\setminus\widehat{I}}e_k (\kappa_k - \varphi_k(t')) \geq \hat{e} \left(\varphi_{\ell}(t)- \varphi_{\ell}(t')\right) + \hat{e} \sum_{i\in \widehat{I}\setminus\widehat{I'}}(\kappa_i -\varphi_i(t)).
\end{align*}
But from (EH) and \eqref{eq:classement_rate}, we also have
\begin{align*}
\sum_{k\in \widehat{I'}\setminus\widehat{I}}e_k (\kappa_k - \varphi_k(t')) < {e_{\ell}} \left(\varphi_{\ell}(t)- \varphi_{\ell}(t')\right) + {e}_{\ell} \sum_{i\in \widehat{I}\setminus\widehat{I'}}(\kappa_i -\varphi_i(t)),
\end{align*}
then 
\begin{align*}
0\geq (\hat{e}-{e_{\ell}}) \left(\varphi_{\ell}(t)- \varphi_{\ell}(t')\right) + (\hat{e}-{e_{\ell}}) \sum_{i\in \widehat{I}\setminus\widehat{I'}}(\kappa_i -\varphi_i(t)),
\end{align*}
which contradicts our assumption. 

\noindent
{\bf (i-b)} If $\displaystyle \sum_{i\in I}\varphi_i(t) <  D(t)$ then, 
for all $i\in I$, $\varphi_i(t)= \kappa_i$ and (EH) is necessarily false. 

\newcommand{\II}{\,\,{I}\!\!\!\!\!\!{I}\,\,\,'}
\subsubsection*{(ii) We examine the situation $I(t') \neq I(t)$}

We add the following shortened notation: $I(t)\cap I'(t) = {\II}$.\\
We break down $I$ and $I'$ into the sets $\II$, $I\backslash I'$ and $I'\backslash I$. We denote by $\widehat{I}$ the set of index $i\in I$ such that $c_i+t e_i = \uup(t)$. In particular, when $j\in I \backslash  \widehat{I}$, then $\varphi_j(t) = \kappa_j$. 

We first derive some generic relations between the emission rates for these. 

Among the indexes in the set $\II$, we observe that at most one index exists (say $\ell$) in the set $\widehat{I}\cap\widehat{I'}$. 
If $j \in \widehat{I}\backslash\widehat{I'}$, if $k\in\widehat{I'}\backslash\widehat{I}$,  then, by the definition of the sets
\begin{align*}
\begin{array}{ll}
c_j + e_j t = c_{\ell} + e_{\ell} t, & \quad c_k + e_k t < c_j + e_j t ,\\
c_j + e_j t' < c_{\ell} + e_{\ell} t', & \quad c_k + e_k t' = c_{\ell} + e_{\ell} t' ,\\
c_j + e_j t' < c_k + e_k t  , & \quad  c_k + e_k t < c_{\ell} + e_{\ell} t ,\\
\end{array}
\end{align*}
from which, we easily deduce that 
\begin{equation}\label{eq:classement_1}
\begin{aligned}
\hat{e}:=   \max\left\{e_j, j\in \II\cap\left( \widehat{I}\backslash\widehat{I'}\right)\right\} <  e_{\ell} < \min\left\{e_k,k\in \II\cap\left(\widehat{I'}\backslash\widehat{I}
\right)\right\}:=\hat{e}'.
\end{aligned}
\end{equation}

\noindent  For $j\in I\backslash I'$ and  $k \in I'\backslash I$, we have
\begin{align*}
c_j + e_j t  &< c_k + e_k t \quad \mbox{ and } \quad 
c_j + e_j t' >  c_k + e_k t' 
\end{align*}
from which, we also easily deduce that 
\begin{align}\label{eq:classement_2}
\max\{e_k, k \in I'\backslash I\} < \min\{e_j, j\in I\backslash I'\}.
\end{align}
For the same $j$ and $k$, for $(\hat{c},\hat{e})$ representative of index in $\II\cap \widehat{I}\setminus \widehat{I'}$, and  $(\hat{c}',\hat{e}')$ representative of index in $\II\cap \widehat{I'}\setminus \widehat{I}$, we  also have 
\begin{align*}
\begin{array}{ll}
c_j + e_j t  &\leq  \hat{c} + \hat{e}  t \\
c_j + e_j t' &>  \hat{c} + \hat{e}  t'
\end{array}
\quad\mbox{ and }\quad 
 \begin{array}{ll}
c_k + e_k t  &>  \hat{c}' + \hat{e}'  t \\
c_k + e_k t' &\leq  \hat{c}' + \hat{e}'  t'
\end{array}
\end{align*}
from which, we  deduce that 
\begin{equation}\label{eq:classement_3}
\begin{aligned}
\min\{e_j, j\in I\backslash I'\} &> (e_{\ell},\hat{e})\vee \max\{e_k, k \in I'\backslash I\}\\
\max\{e_k, k \in I'\backslash I\} &< (e_{\ell}, \hat{e}')\wedge \min\{e_j, j\in I\backslash I'\}.
\end{aligned}
\end{equation}

We divide the analysis in cases. In the first one the demand is fully satisfied for the price $\pelec(t)$. 

\medskip\noindent
{\bf (ii-a)} If $\sum_{i\in I}\varphi_I(t) = D(\pelec(t))$, 

\renewcommand{\theequation}{DC}
\begin{align}
\sum_{i\in I\backslash I' }\varphi_I + \sum_{i\in \II }\varphi_i(t) = D(\pelec(t)) \geq  D(\pelec(t'))\geq \sum_{i\in \II}\varphi_i(t') + \sum_{i\in I'\backslash  I}\varphi_i(t') , 
\end{align}
\renewcommand{\theequation}{EH}
\begin{align}
\sum_{i\in I\backslash I' }\varphi_i(t)e_i + \sum_{i\in \II}\varphi_i(t) e_i  <  \sum_{i\in \II}\varphi_i(t')e_i  + \sum_{i\in I'\backslash I }\varphi_i(t) e_i .
\end{align}
\renewcommand{\theequation}{\thesection.\arabic{equation}}
We must then examine the following two subcases, relative to the situations where the demand is satisfied or not at the price $\pelec(t')$.

\medskip\noindent {\bf (ii-a-1)} If $\sum_{i\in I'}\varphi_i(t') <  D(\pelec(t'))$, then $\varphi_i(t') = \kappa_i$ for all $i\in I'$ and 
\renewcommand{\theequation}{DC}
\begin{align}
\sum_{j\in I\backslash I' }\varphi_j(t) + \sum_{i\in \II }\varphi_i(t)  >  \sum_{i\in \II}\kappa_i + \sum_{k\in I'\backslash  I}\kappa_k,
\end{align}
\renewcommand{\theequation}{EH}
\begin{align}
\sum_{j\in I\backslash I' }\varphi_j(t)e_j + \sum_{i\in \II}\varphi_i(t) e_i  <  \sum_{i\in \II}\kappa_i e_i  + \sum_{k\in I'\backslash I }\kappa_k e_k .
\end{align}
\renewcommand{\theequation}{\thesection.\arabic{equation}}
As $\varphi_i(t) = \kappa_i$ when $i\in (I\backslash \widehat{I})\cap \II$, we can simplify the two sides of (DC) and (EH) by the sum over $(I\backslash \widehat{I})\cap \II$. The remaining part of $\II$ is $\{\ell \}  \cup \left(\widehat{I}\backslash\widehat{I'} \cap \II \right)$: 
\renewcommand{\theequation}{DC}
\begin{align}
\sum_{j\in I\backslash I' }\varphi_j(t) 
+ \varphi_{\ell} 
+ \sum_{i\in\widehat{I}\backslash\widehat{I'} \cap \II }\varphi_i(t)  
>  
\kappa_{\ell} 
+ \sum_{i\in\widehat{I}\backslash\widehat{I'} \cap \II }\kappa_i 
+ \sum_{k\in I'\backslash  I}\kappa_k ,
\end{align}
\renewcommand{\theequation}{EH}
\begin{align}
\sum_{j\in I\backslash I' }e_j\varphi_j(t) 
+ e_{\ell} \varphi_{\ell} 
+ \sum_{i\in\widehat{I}\backslash\widehat{I'} \cap \II }e_i \varphi_i(t)  
< 
e_{\ell} \kappa_{\ell} 
+ \sum_{i\in\widehat{I}\backslash\widehat{I'} \cap \II }e_i \kappa_i 
+ \sum_{k\in I'\backslash  I}e_k \kappa_k.
\end{align}
\renewcommand{\theequation}{\thesection.\arabic{equation}}
Then we multiply (DC) by $\bar{e} := (e_{\ell},\hat{e})\vee \max\{e_k, k \in I'\backslash I\}$, and we obtain by \eqref{eq:classement_3}
\begin{align*}
\sum_{j\in I\backslash I' }e_j \varphi_j(t) 
+\bar{e} \varphi_{\ell} 
+ \bar{e}\sum_{i\in\widehat{I}\backslash\widehat{I'} \cap \II }\varphi_i(t)  
>  
\bar{e} \kappa_{\ell} 
+\bar{e}  \sum_{i\in\widehat{I}\backslash\widehat{I'} \cap \II }\kappa_i 
+ \sum_{k\in I'\backslash  I}e_k \kappa_k.
\end{align*}
We subtract with (EH) : 
\begin{align*}
(\bar{e}-e_{\ell})  \varphi_{\ell} 
+ \sum_{i\in\widehat{I}\backslash\widehat{I'} \cap \II } (\bar{e} -e_i) \varphi_i(t)  
>  
(\bar{e} -e_{\ell})\kappa_{\ell} 
+ \sum_{i\in\widehat{I}\backslash\widehat{I'} \cap \II }(\bar{e} -e_i)  \kappa_i.
\end{align*}
But $\bar{e}\geq e_{\ell}$ when $\ell$ exists, and $\bar{e}\geq \hat{e}\geq e_i$ for  $i\in\widehat{I}\backslash\widehat{I'} \cap \II$. So we obtain our contradiction. 

\medskip\noindent {\bf (ii-a-2)} If $\sum_{i\in I'}\varphi_i(t') = D(\pelec(t'))$, then 
\renewcommand{\theequation}{DC}
\begin{align}
\sum_{j\in I\backslash I' }\varphi_j(t) 
+ \sum_{i\in \II }\varphi_i(t)  
>  \sum_{i\in \II}\varphi_i(t') 
+ \sum_{k\in I'\backslash  I}\varphi_k(t'), 
\end{align}
\renewcommand{\theequation}{EH}
\begin{align}
\sum_{j\in I\backslash I' }\varphi_j(t)e_j + \sum_{i\in \II}\varphi_i(t) e_i  <  \sum_{i\in \II}\varphi_i(t') e_i  + \sum_{k\in I'\backslash I }\varphi_k(t') e_k.
\end{align}
We decompose $I\backslash I' = \left(I\backslash(I'\cup \widehat{I}) \right) \cup \widehat{I}\backslash I'$ and $I'\backslash I = \left(I'\backslash (I\cup \widehat{I'})\right) \cup \widehat{I'}\backslash  I$: 
\renewcommand{\theequation}{DC}
\begin{align}
& \sum_{j\in I\backslash (I'\cup \widehat{I}) }\kappa_j 
+ \sum_{j\in \widehat{I}\backslash I'}\varphi_j(t) 
+ \sum_{i\in \II }\varphi_i(t)  
 >  \sum_{i\in \II}\varphi_i(t') 
+ \sum_{k\in \widehat{I'}\backslash  I}\varphi_k(t') 
+ \sum_{k\in I'\backslash (I\cup \widehat{I'}) }\kappa_k,
\end{align}
\renewcommand{\theequation}{EH}
\begin{align}
& \sum_{j\in I\backslash (I'\cup \widehat{I}) }e_j \kappa_j 
+ \sum_{j\in \widehat{I}\backslash I'}e_j \varphi_j(t) 
+ \sum_{i\in \II }e_i \varphi_i(t)  
<  \sum_{i\in \II}e_i \varphi_i(t') 
+ \sum_{k\in \widehat{I'}\backslash  I}e_k \varphi_k(t') 
+ \sum_{k\in I'\backslash (I\cup \widehat{I'})}e_k \kappa_k.
\end{align}
We also break down the set $\II = (I\cap I')$: 
\begin{align*}
\II  =  & \left(\II\cap\{\ell\}\right)  \cup \left(\II\cap\widehat{I}\backslash\widehat{I'}\right)\cup \left(\II\cap \widehat{I'}\backslash\widehat{I}\right)\cup \left({I}\backslash\widehat{I} \cap{I'}\backslash\widehat{I'}) \right).
\end{align*}
\renewcommand{\theequation}{DC}
\begin{align}
& \sum_{j\in I\backslash (I'\cup \widehat{I})) }\kappa_j 
+ \sum_{j\in \widehat{I}\backslash I'}\varphi_j(t) 
+\varphi_{\ell}(t) 
+ \sum_{i\in \widehat{I}\backslash\widehat{I'} \cap \II} \varphi_i(t) 
+ \sum_{i\in \widehat{I'}\backslash\widehat{I} \cap \II} \varphi_i(t)  \nonumber \\
& >  
\varphi_{\ell}(t') 
+ \sum_{i\in \widehat{I}\backslash\widehat{I'} \cap \II} \varphi_i(t') 
+ \sum_{i\in \widehat{I'}\backslash\widehat{I} \cap \II} \varphi_i(t')
+ \sum_{k\in \widehat{I'}\backslash  I}\varphi_k(t') 
+ \sum_{k\in I'\backslash (I\cup \widehat{I'})}\kappa_k ,
\end{align}
\renewcommand{\theequation}{EH}
\begin{align}
& \sum_{j\in I\backslash (I'\cup \widehat{I}) }e_j \kappa_j 
+ \sum_{j\in \widehat{I}\backslash I'}e_j \varphi_j(t) 
+e_{\ell} \varphi_{\ell}(t) 
+ \sum_{i\in \widehat{I}\backslash\widehat{I'} \cap \II} e_i \varphi_i(t) 
+ \sum_{i\in \widehat{I'}\backslash\widehat{I} \cap \II} e_i \varphi_i(t) \nonumber\\
& <  
e_{\ell} \varphi_{\ell}(t') 
+ \sum_{i\in \widehat{I}\backslash\widehat{I'} \cap \II} e_i \varphi_i(t') 
+ \sum_{i\in \widehat{I'}\backslash\widehat{I} \cap \II} e_i \varphi_i(t')
+ \sum_{k\in \widehat{I'}\backslash  I}e_k \varphi_k(t') 
+ \sum_{k\in I'\backslash (I\cup \widehat{I'})}e_k \kappa_k . 
\end{align}
\renewcommand{\theequation}{\thesection.\arabic{equation}}
For index $i$ in the last subset $({I}\backslash\widehat{I} \cap{I'}\backslash\widehat{I'})$, we have $\varphi_i(t) = \kappa_i$ and $\varphi_i(t') = \kappa_i$,  so we simplify (DC) and (EH) from this last subset. Thus, 
\renewcommand{\theequation}{DC}
\begin{align}
& \sum_{j\in I\backslash (I'\cup \widehat{I}) }\kappa_j 
+ \sum_{j\in \widehat{I}\backslash I'}\varphi_j(t) 
+\varphi_{\ell}(t) 
+ \sum_{i\in \widehat{I}\backslash\widehat{I'} \cap \II} \varphi_i(t) 
+ \sum_{i\in \widehat{I'}\backslash\widehat{I} \cap \II} \kappa_i \nonumber\\
&\qquad >  
\varphi_{\ell}(t') 
+ \sum_{i\in \widehat{I}\backslash\widehat{I'} \cap \II} \kappa_i 
+ \sum_{i\in \widehat{I'}\backslash\widehat{I} \cap \II} \varphi_i(t')
+ \sum_{k\in \widehat{I'}\backslash  I}\varphi_k(t') 
+ \sum_{k\in I'\backslash (I\cup \widehat{I'}) }\kappa_k,
\end{align}
\renewcommand{\theequation}{EH}
\begin{align}
& \sum_{j\in I\backslash (I'\cup \widehat{I}) }e_j \kappa_j 
+ \sum_{j\in \widehat{I}\backslash I'}e_j \varphi_j(t) 
+e_{\ell} \varphi_{\ell}(t) 
+ \sum_{i\in \widehat{I}\backslash\widehat{I'} \cap \II} e_i \varphi_i(t) 
+ \sum_{i\in \widehat{I'}\backslash\widehat{I} \cap \II} e_i \kappa_i \nonumber \\
&\qquad  <  
e_{\ell} \varphi_{\ell}(t') 
+ \sum_{i\in \widehat{I}\backslash\widehat{I'} \cap \II} e_i \kappa_i
+ \sum_{i\in \widehat{I'}\backslash\widehat{I} \cap \II} e_i \varphi_i(t')
+ \sum_{k\in \widehat{I'}\backslash  I}e_k \varphi_k(t') 
+ \sum_{k\in I'\backslash (I\cup \widehat{I'} }e_k \kappa_k. 
\end{align}
\renewcommand{\theequation}{\thesection.\arabic{equation}}
We multiply (DC) by $\bar{e}:=(e_{\ell},\hat{e})\vee \max\{e_k, k \in I'\backslash I\}$ , we get by \eqref{eq:classement_3}
\begin{align*}
& \sum_{j\in I\backslash (I'\cup \widehat{I}) }e_j \kappa_j 
+ \sum_{j\in \widehat{I}\backslash I'}e_j \varphi_j(t) 
+\bar{e} \varphi_{\ell}(t) 
+ \bar{e}\sum_{i\in \widehat{I}\backslash\widehat{I'} \cap \II} \varphi_i(t) 
+ \bar{e}\sum_{i\in \widehat{I'}\backslash\widehat{I} \cap \II} \kappa_i\\
& >  
\bar{e} \varphi_{\ell}(t') 
+ \bar{e} \sum_{i\in \widehat{I}\backslash\widehat{I'} \cap \II} \kappa_i 
+ \bar{e} \sum_{i\in \widehat{I'}\backslash\widehat{I} \cap \II} \varphi_i(t')
+ \sum_{k\in \widehat{I'}\backslash  I}e_k \varphi_k(t') 
+ \sum_{k\in I'\backslash (I\cup \widehat{I'} }e_k \kappa_k.
\end{align*}
We subtract (EH) 
\begin{align*}
&(\bar{e} - e_{\ell})  \varphi_{\ell}(t) 
+ \sum_{i\in \widehat{I}\backslash\widehat{I'} \cap \II} (\bar{e} - e_i) \varphi_i(t) 
+ \sum_{i\in \widehat{I'}\backslash\widehat{I} \cap \II} (\bar{e} - e_i)  \kappa_i\\
& >  
(\bar{e} - e_{\ell})  \varphi_{\ell}(t') 
+ \sum_{i\in \widehat{I}\backslash\widehat{I'} \cap \II} (\bar{e} - e_i)  \kappa_i 
+ \sum_{i\in \widehat{I'}\backslash\widehat{I} \cap \II} (\bar{e} - e_i)  \varphi_i(t').
\end{align*}
We arrange the terms 
\begin{align*}
&(\bar{e} - e_{\ell})  \varphi_{\ell}(t) 
+ \sum_{i\in \widehat{I}\backslash\widehat{I'} \cap \II} (\bar{e} - e_i) \varphi_i(t) 
+ \sum_{i\in \widehat{I'}\backslash\widehat{I} \cap \II} (\bar{e} - e_i)  \kappa_i\\
& >  
(\bar{e} - e_{\ell})  \varphi_{\ell}(t') 
+ \sum_{i\in \widehat{I}\backslash\widehat{I'} \cap \II} (\bar{e} - e_i)  \kappa_i 
+ \sum_{i\in \widehat{I'}\backslash\widehat{I} \cap \II} (\bar{e} - e_i)  \varphi_i(t').
\end{align*}
If $\ell$ exists, then $\bar{e} = e_{\ell}$ and 
\begin{equation}\label{eq:tmp_for_iib2_with_l}
\begin{aligned}
 \sum_{i\in \widehat{I'}\backslash\widehat{I} \cap \II} (e_{\ell} - e_i)  \left(\kappa_i - \varphi_i(t')\right)
 & >  
\sum_{i\in \widehat{I}\backslash\widehat{I'} \cap \II} (e_{\ell} - e_i)  \left( \kappa_i -  \varphi_i(t) \right), \\
 \sum_{i\in \widehat{I'}\backslash\widehat{I} \cap \II} (e_{\ell} - \hat{e}')  \left(\kappa_i - \varphi_i(t')\right)
&  >  
\sum_{i\in \widehat{I}\backslash\widehat{I'} \cap \II} (e_{\ell} - \hat{e})  \left( \kappa_i -  \varphi_i(t) \right). 
\end{aligned}
\end{equation}
But $\hat{e} < e_{\ell} <  \hat{e}'$, and the contradiction follows. \\
If $\ell$ does not exist, then $\bar{e} = \hat{e}\vee \max\{e_k, k \in I'\backslash I\}$ 
\begin{equation}\label{eq:tmp_for_iib2_without_l}
\begin{aligned}
 \sum_{i\in \widehat{I'}\backslash\widehat{I} \cap \II} (\bar{e} - e_i)  \left(\kappa_i - \varphi_i(t')\right)
 & >  
\sum_{i\in \widehat{I}\backslash\widehat{I'} \cap \II} (\bar{e} - e_i)  \left( \kappa_i -  \varphi_i(t) \right), \\
 \sum_{i\in \widehat{I'}\backslash\widehat{I} \cap \II} (\bar{e} - \hat{e}')  \left(\kappa_i - \varphi_i(t')\right)
&  >  
\sum_{i\in \widehat{I}\backslash\widehat{I'} \cap \II} (\bar{e} - \hat{e})  \left( \kappa_i -  \varphi_i(t) \right) .
\end{aligned}
\end{equation}
But $\max\{e_k, k \in I'\backslash I\} < \hat{e}'$, and the contradiction follows. 

\medskip\noindent
{\bf (ii-b)} If $\sum_{i\in I}\varphi_i(t) <  D(\pelec(t))$ then 
for all $i\in I$, $\varphi_i(t)= \kappa_i$. 

\noindent {\bf (ii-b 1)} If $\sum_{i\in I'}\varphi_i(t') <  D(\pelec(t'))$, then $\varphi_i(t') = \kappa_i$ for all $i\in I'$. Moreover, we have that $\Off(t,\uup(t) )\geq D(\uup(t)) + \varepsilon) \geq D(\uup(t')) > \Off(t',\uup(t'))$ and (DC)-(EH) becomes
\renewcommand{\theequation}{DC}
\begin{align}
\sum_{j\in I\backslash I' }\kappa_j  >  \sum_{k\in I'\backslash  I}\kappa_k,
\end{align}
\renewcommand{\theequation}{EH}
\begin{align}
\sum_{j\in I\backslash I' }e_j \kappa_j  <  \sum_{k\in I'\backslash I }e_k \kappa_k. 
\end{align}
\renewcommand{\theequation}{\thesection.\arabic{equation}}
Then, we multiply (DC) by $\min\{e_j ;j\in I\backslash I'\} \geq \max\{e_k; k \in I'\backslash I\}$, and we obtain a contradiction with (EH). 

\medskip\noindent {\bf (ii-b-2)} If $\sum_{i\in I'}\varphi_i(t') = D(\pelec(t'))$, we go back to the analysis of the case {\bf (ii-a-2)}, with the main difference that all quantities $\varphi_i(t)$ are now equal to $\kappa_i$. We go to inequalities  \eqref{eq:tmp_for_iib2_with_l} and \eqref{eq:tmp_for_iib2_without_l} which are simplified as the right-had sides are now zero. The contradiction follows with the same arguments.

\end{document}